\documentclass[12pt, reqno]{amsart}

\makeatletter
\g@addto@macro{\endabstract}{\@setabstract}
\makeatother

\usepackage{graphics, stackrel}
\usepackage{amsmath, amssymb, amsthm}
\usepackage{graphicx}
\usepackage{verbatim}
\usepackage{amsfonts}

\usepackage{natbib}

\usepackage{booktabs}

\usepackage{fancyvrb}
\usepackage{color}
\usepackage{mdwlist}
\usepackage{mathtools}
\usepackage{accents}
\usepackage[normalem]{ulem}

\usepackage{enumitem}
\setlist[enumerate]{itemsep=2pt,topsep=3pt}
\setlist[itemize]{itemsep=2pt,topsep=3pt}
\setlist[enumerate,1]{label=(\alph*)}

\usepackage{caption}
\usepackage{subcaption}

\usepackage{graphicx}
\usepackage{amsmath}
\makeatletter
\newcommand{\distas}[1]{\mathbin{\overset{#1}{\kern\z@\sim}}}%
\newsavebox{\mybox}\newsavebox{\mysim}
\newcommand{\distras}[1]{%
  \savebox{\mybox}{\hbox{\kern3pt$\scriptstyle#1$\kern3pt}}%
  \savebox{\mysim}{\hbox{$\sim$}}%
  \mathbin{\overset{#1}{\kern\z@\resizebox{\wd\mybox}{\ht\mysim}{$\sim$}}}%
}
\makeatother

\usepackage{mathrsfs}
\usepackage{bbm}

\usepackage[left=1.1in, right=1.1in, top=0.85in, bottom=1.05in, includehead, includefoot]{geometry}

\renewcommand{\leq}{\leqslant}
\renewcommand{\geq}{\geqslant}

\newcommand\blfootnote[1]{%
  \begingroup
  \renewcommand\thefootnote{}\footnote{#1}%
  \addtocounter{footnote}{-1}%
  \endgroup
}

\theoremstyle{plain}
\newtheorem{theorem}{Theorem}[section]

\newtheorem{lemma}[theorem]{Lemma}
\newtheorem{proposition}[theorem]{Proposition}

\theoremstyle{definition}

\newtheorem{assumption}{Assumption}[section]

\newcommand{\setntn}[2]{ \{ #1 : #2 \} }

\newcommand{\1}{\mathbbm 1}
\newcommand{\la}{\langle}
\newcommand{\ra}{\rangle}

\renewcommand{\phi}{\varphi}
\renewcommand{\epsilon}{\varepsilon}

\newcommand{\cC}{\mathscr C}

\newcommand{\rR}{\mathcal R}

\newcommand{\RR}{\mathbbm R}

\newcommand{\NN}{\mathbbm N}

\newcommand{\EE}{\mathbbm E}

\newcommand{\XX}{\mathbbm X}
\newcommand{\YY}{\mathbbm Y}

\newcommand{\lrM}{\mathcal M_C}
\newcommand{\olrM}{\overline{\mathcal M}_C}

\newcommand{\diff}{\mathrm d}

\usepackage{xr-hyper}
\usepackage[citecolor=blue, colorlinks=true, linkcolor=blue]{hyperref}

\begin{document}

\thispagestyle{empty}

\title{}

\begin{center}
  \LARGE
    Necessary and Sufficient Conditions for Existence and Uniqueness of  Recursive Utilities\blfootnote{The authors thank
        Anmol Bhandari, Tim Christensen, Ippei Fujiwara, Jinill Kim,
        Daisuke Oyama, and Guanlong Ren for useful comments and suggestions.
        Special thanks are due to Miros\l{}awa Zima for her valuable input on
    local spectral radius conditions.  The second author gratefully
acknowledges financial support from ARC grant FT160100423.  The views
expressed herein are those of the authors and not necessarily those of the
Federal Reserve Bank of Minneapolis or the Federal Reserve System.}

    \vspace{1em}

  \large
  Jaroslav Borovi\v{c}ka\textsuperscript{a}
  and John Stachurski\textsuperscript{b} \par \bigskip

  \small
  \textsuperscript{a} New York University, Federal Reserve Bank of Minneapolis and NBER \par
  \textsuperscript{b} Research School of Economics, Australian National University \bigskip

  \normalsize
  \today
\end{center}





\begin{abstract}
    We obtain exact necessary and sufficient conditions for existence and
    uniqueness of solutions of a class of homothetic recursive utility models
    postulated by \cite{epstein1989}.  The conditions center on a single test value
    with a natural economic interpretation.  The test sheds light on the
    relationship between valuation of cash flows, impatience, risk adjustment and intertemporal
    substitution of consumption.  We propose two methods to compute the test
    value when an analytical solution is not available.  Several applications
    are provided.

    \vspace{1em}

\end{abstract}

\clearpage
\thispagestyle{empty}
{\bf Disclosure statement for Jaroslav Borovi\v{c}ka}

The views expressed in the paper are my own and not necessarily those of the
Federal Reserve Bank of Minneapolis or the Federal Reserve System. I have nothing else to disclose.

\clearpage
\thispagestyle{empty}
{\bf Disclosure statement for John Stachurski}

I have nothing to disclose.

\clearpage
\setcounter{page}{1}
\section{Introduction}

Recursive preference models such as those discussed in
\cite{koopmans1960stationary}, \cite{epstein1989} and \cite{Weil1990}
play an important role in macroeconomic and financial
modeling.  For example, the long-run risk models analyzed in
\cite{bansal2004risks}, \cite{hansen2008consumption},
\cite{bansal2012empirical} and \cite{schorfheide2018identifying} have
employed such preferences in discrete time infinite horizon settings
with a variety of consumption path specifications to
help resolve long-standing empirical puzzles identified in the literature.

In recursive utility models, the lifetime value of a consumption stream from
a given point in time is expressed as the solution to a nonlinear
forward-looking equation.  While this representation is convenient and intuitive, it
can also be vacuous, in the sense that no finite solution to the forward
looking recursion exists.  Moreover, even when a solution is found, this
solution lacks predictive content unless some form of uniqueness can
also be established.  In general, identifying restrictions that imply existence
and uniqueness of a solution for an empirically relevant class of consumption
streams is challenging.

The aim of the present paper is to obtain existence and uniqueness results
that are as tight as possible in a range of empirically plausible settings,
while restricting attention to practical conditions that can be tested in
applied work.  To this end, we provide conditions for existence and uniqueness
of solutions to the class of homothetic preferences studied in
\cite{epstein1989}, while admitting both stationary and nonstationary
consumption paths.  These conditions are both necessary and sufficient, and
hence as tight as possible in the setting we consider.  In particular, if the
conditions hold then a unique, globally attracting solution exists, while if
not then no finite solution exists. Existence of a finite solution is
equivalent to the existence of a finite wealth-consumption ratio, a central
object of interest in asset pricing.

To give more detail on that setting, let preferences be defined recursively by
\begin{equation}
    \label{eq:agg}
    V_t = \left[
            (1 - \beta) C_t^{1-1/\psi}
            + \beta \left\{ \rR_t \left(V_{t+1}
            \right) \right\}^{1-1/\psi}
          \right]^{1/(1-1/\psi)},
\end{equation}
where $\{C_t\}$ is a consumption path, $V_t$ is the utility value of the path
extending on from time $t$ and $\rR_t$ is the Kreps--Porteus
certainty equivalent operator
\begin{equation}
    \label{eq:ce}
    \rR_t(V_{t+1})
    = ( \EE_t  V^{1-\gamma}_{t+1} )^{1/(1-\gamma)}.
\end{equation}
The parameter $\beta \in (0, 1)$ is a time discount factor, while $\gamma
\not=1$
governs risk aversion and $\psi \not=1$ is the elasticity of intertemporal
substitution.  We take
the consumption stream as given and seek a solution for normalized utility
$V_t / C_t$.

The first step in our approach is to associate to each consumption process the risk-adjusted
long-run mean consumption growth rate
\begin{equation}
    \label{eq:drc}
    \lrM
    := \lim_{n \to \infty }
        \left[
            \rR \left( \frac{C_n}{C_0} \right)
            \right]^{1/n} ,
\end{equation}
where $\rR$ is the unconditional version of the Kreps--Porteus certainty
equivalent operator.  Beginning with the case where the state vector driving
the conditional distribution of consumption growth takes values in a compact
set---which is where the sharpest results can be obtained---we show that a
unique solution exists if and only if $\Lambda < 1$, where
\begin{equation}
    \label{eq:kc}
    \Lambda := \beta \, \lrM^{1-1/\psi}.
\end{equation}
Under the same compactness restriction, we also show that the condition
$\Lambda < 1$ is both necessary and sufficient for global convergence of
successive approximations associated with a natural fixed point mapping. In
fact our results establish that convergence of successive approximations
itself implies that a unique solution exists, and that the limit produced
through this process is equal to the solution.  Furthermore, we prove that
when the condition $\Lambda < 1$ fails, not only does existence and uniqueness
of a solution fail, but existence fails specifically.

The value $\Lambda$ represents the risk-adjusted long-term consumption growth
rate modified by
impatience and the intertemporal substitutability
of consumption.  Despite the fact that the preference recursion \eqref{eq:agg} intertwines the
contributions of impatience, intratemporal risk aversion and intertemporal elasticity of
substitution to value, the condition $\Lambda < 1$ effectively separates these forces.
Details of the consumption growth process, such as its
persistence or higher moments of its innovations, matter
only through the long-run distribution of consumption growth encoded in $\lrM$.  Additional discussion of the intuition behind condition
\eqref{eq:kc} is provided in section~\ref{s:ez} and the applications.

In addition to the preceding results, we use a local spectral theorem to show that
\begin{equation}
    \label{eq:deftheta}
    \lrM = r(K)^{1/(1-\gamma)},
\end{equation}
where $r(K)$ is the spectral radius of a valuation operator $K$ determined by the
primitives and clarified below.
This result is useful on two levels. First, spectral radii and dominant
eigenfunctions associated with valuation operators have
increasingly been used to understand long-run risks and long-run values in
macroeconomic and financial applications by inducing a decomposition of the
stochastic discount factor (see, e.g., \cite{alvarez2005using, hansen2009,
qin2017long, christensen2017nonparametric}). The identification
in \eqref{eq:deftheta} allows us to connect to and draw
insights from this literature.  Second, on a computational level, when the
state space for the state process is
finite, the valuation operator $K$ is just a matrix, and the spectral radius
is easily computed.  From this one can compute the test statistic $\Lambda$
via \eqref{eq:deftheta}.  When the state space is not
finite, one can still implement this idea after discretization.

When the state space is high dimensional, accurate discretization is
nontrivial and calculation of the spectral radius becomes computationally
expensive.  For these scenarios, we propose instead a Monte Carlo method to
calculate the test value $\Lambda$, based around the idea of simulating
consumption paths from a given specification and calculating the risk-adjusted
expectation on the right hand side of \eqref{eq:drc} by averaging over these
paths.  This approach is straightforward to implement and relatively
insensitive to the dimension of the state space.  Another advantage is that
the routine is easily be parallelized by simulating independent consumption
paths along multiple execution threads.

All of the theoretical results on existence, uniqueness and convergence of
successive approximations discussed above are stated in the context of a
compact-valued state process, which drives the persistent component of
consumption growth.  In this setting we apply a fixed point theorem due to \cite{du1990fixed},
which extends to abstract vector space the idea that an increasing concave
function $f$ from $\RR_+$ to itself has at most one strictly positive fixed
point---and at least one such fixed point if the slope conditions $f'(0) > 1$
and $f'(\infty) < 1$ are satisfied.  In the case of the valuation problem
considered in this paper, monotonicity and concavity are inherent in the
preference specification, while the condition $\Lambda < 1$ is the key to the
slope conditions.  The arguments required for the last step are nontrivial and
compactness of the state space plays a significant role.

At the same time, we do provide some guidance on the implications of the
condition $\Lambda < 1$ when the state space is unbounded.  In this setting we
show that, under an auxiliary technical restriction related to compactness of
iterates of the operator $K$, the condition $\Lambda < 1$ is again necessary
and sufficient for existence of a solution.  As before, convergence of
successive approximations to a given function implies that the limiting
function is a solution, and this can only occur when $\Lambda < 1$.  The
identification in \eqref{eq:deftheta} continues to be valid when these
conditions hold, and the Monte Carlo method described above can still be
applied.  The proofs use an approximation argument that bootstraps previously
stated results from the compact case.

We provide a series of applications of the numerical and theoretical results
in section~\ref{s:a}. We start with the model of a trend-stationary
consumption process from \cite{alvarez2005using}, in which we can characterize
our condition analytically. The result reveals that transitory uncertainty in
the consumption process, such as stationary fluctuations around a
deterministic time trend, are immaterial for the existence of the continuation
value.

Next, we focus on a frequently used case in which the dynamics of conditional
moments of consumption growth are encoded using a Markov chain. Specifically,
we use the calibration from \cite{johannes2016learning} and consider two
information structures, one in which the agent observes the realizations of
the Markov chain and another in which the agent must learn about the
underlying state from consumption growth realizations that reveal the state
only imperfectly. In the former case, given the simple structure of the state
space, evaluating the test condition involves computing the spectral radius in
(\ref{eq:deftheta}) as the largest eigenvalue of a small matrix. In the latter
case, the state space is continuous and encodes subjective probabilities of
unobserved states of the Markov chain, updated using Bayes rule. Despite the
fact that the state space and transition dynamics are now much more complex,
our results reveal that the continuation value exists exactly for the same set
of parameters as in the full information case. The underlying reason is the
transitory role of state uncertainty for the conditional distribution of
future consumption growth that has no impact on the value of $\lrM$ in
(\ref{eq:drc}).

Our main quantitative application is the long-run
risk model specified in \cite{schorfheide2018identifying}.  Existence of finite
continuation values is a nontrivial issue in such models because
(i) discounting is extremely small in order to raise the
importance of persistent risk components, pushing them towards the boundary
between stability and instability, and (ii) the state dynamics are nonlinear
and relatively high dimensional.  Using numerical methods, we show that the
condition $\Lambda < 1$ is satisfied 
with arbitrarily small amounts of truncation and almost certainly satisfied in 
the original model without truncation.  We also provide several robustness checks that relate
alternative approximation procedures and the impact of truncation.

Finally, we show that our results can also be applied to 
production economies where consumption is endogenously determined. In many
applications, consumption is cointegrated with an exogenously specified
process that drives uncertainty in the model. Since transitory fluctuations in
the consumption process are irrelevant, the risk-adjusted long-run growth rate
can be directly computed using the exogenous driving process, without the
knowledge of specific details of the consumption process.

Regarding existing literature, sufficient conditions for existence and
uniqueness of recursive utilities were provided by \cite{epstein1989} and
\cite{marinacci2010}.  These conditions require a finite bound $B_c$ on
consumption growth $C_{t+1} / C_t$ that holds asymptotically with probability
one.  As a result, they cannot be applied to many recent specifications of the
consumption processes, such as the long-run risk specification given in
section~\ref{ss:lrr}, as consumption growth in those settings is unbounded
above.  Even if a finite bound is obtained by truncation of the shocks, we
show that the resulting conditions are always stricter than the ones presented
in this paper, and typically far too strict for realistic parameterizations.  This is
due to the fact that probability one bounds restrict utility uniformly along
every future consumption trajectory, while the results in this paper consider
what happens on average across all paths.  In other words, our results are sharper because
recursive utility specifications, while nonlinear, are still defined using
integration over future continuation values. Conversely, focusing only on the
upper tail of the consumption growth process leads to excessively tight
stability conditions.

Another condition for existence of recursive utilities can be found in
\cite{alvarez2005using}, which focuses on the case where consumption has a
deterministic time trend.  Our condition $\Lambda < 1$ is also weaker than
their condition, as shown in section~\ref{s:pob}.  The intuition behind this
is that \cite{alvarez2005using} use a fixed point argument that requires
contraction in one step.  In contrast, the restriction $\Lambda < 1$ is an
asymptotic condition that ignores short-run fluctuations in consumption.

Also related is \cite{hansen2012recursive}, who study Epstein--Zin
utility models with unbounded consumption growth.  Their approach is to
connect the solution to the Epstein--Zin utility recursion and the
Perron--Frobenius eigenvalue problem associated with a linear operator,
denoted in their paper by $\mathbb T$, that is proportional to the operator
$K$ discussed above.  Consumption growth is a function of an unbounded
exogenous state process.  In this setting they show that a solution exists
when a joint restriction holds on the spectral radius of $\mathbb T$ and the
preference parameters, along with integrability conditions on the
eigenfunctions of $\mathbb T$ already mentioned.  They also obtain a
uniqueness result for some parameter values (although not the most empirically
relevant ones).

For the case where $\XX$ is compact, our approach has the following
advantages: First, we obtain uniqueness of the solution for all
parameterizations.  Second, we show our conditions are necessary as well as
sufficient, both for existence and for uniqueness.  Third, we obtain a
globally convergent method of computation, and show that it converges if and
only if a solution exists.  Fourth, we provide multiple representations of the
test value, strengthening the economic interpretation, as well as a method of
computation that can be applied in high dimensional settings.  Fifth, we avoid
the auxiliary conditions in \cite{hansen2012recursive} involving integrability
restrictions on the eigenfunctions of the operator $\mathbb T$, which means
that all our conditions are straightforward to test in applied settings. For
the case where $\XX$ is not compact, our results also serve to augment those
of \cite{hansen2012recursive} by showing that the condition $\Lambda < 1$ is
necessary as well as sufficient for existence of a solution.

In another related study, \cite{guo2016} consider an extension to the
Epstein--Zin recursive utility model that includes utility measures for
investment gains and losses.  As a part of that study they obtain results for
existence, uniqueness and convergence of solutions to Epstein--Zin recursive
utility models with consumption specifications analogous to those in
\cite{hansen2012recursive}, except that the state space is restricted to be
finite.  In comparison, we allow for the state space to be countably or
uncountably infinite and we establish not just sufficiency but also necessity.


The paper is structured as follows:  Section~\ref{s:cp} considers
the risk-adjusted long-run mean consumption growth rate in more depth.
Section~\ref{s:ez} states our main
results. Section~\ref{s:a} discusses applications.  Section~\ref{s:pob}
contrasts our results with alternative sufficient conditions in the previous
literature.  Section~\ref{s:ub} treats the unbounded case and
section~\ref{s:c} concludes.  All proofs are deferred to the
appendix.\footnote{The repository
\texttt{https://github.com/jstac/recursive\_utility\_code} contains code that
replicates all of our numerical results.}

\section{Consumption Paths and Risk-Adjusted Growth}

\label{s:cp}

Before stating our main results, we introduce our baseline model for
consumption paths and address an important issue: Since $\Lambda  = \beta \,
\lrM^{1-1/\psi}$, the practicality of our condition $\Lambda < 1$ depends
on the ability to accurately evaluate the risk-adjusted long-run mean
consumption growth rate $\lrM$, as defined in \eqref{eq:drc}.  For some
specifications of the consumption path, an analytical expression for $\lrM$
exists.  For others, however, no such expression can be obtained.  In this
second case, we must turn to numerical methods to evaluate
$\lrM$.  This section  discusses two methods to compute $\lrM$ numerically.
The main aims of this section are to (a) build intuition on
$\lrM$ by treating some relatively simple cases and (b) provide
evidence affirming that $\lrM$ and hence $\Lambda$ can be evaluated
sufficiently accurately even when no closed form solution exists.

\subsection{Consumption Paths}

As in \cite{hansen2012recursive}, we suppose that consumption growth has
the generic specification
\begin{equation}
    \label{eq:kappa}
    \ln (C_{t+1}/ C_t) = \kappa(X_t, X_{t+1}, \epsilon_{t+1}),
\end{equation}
where $\kappa$ is a continuous function, $\{X_t\}$ is an exogenous
state process and $\{\epsilon_t\}$ is an {\sc iid} innovation process
supported on $\YY \subset \RR^k$ and independent of $\{X_t\}$.
The state process is assumed to be stationary and
Markov, taking values in a subset $\XX$ of $\RR^n$.  The unconditional density
of each $X_t$ is denoted by $\pi$. The function $q(x, \cdot)$ represents the
conditional density of $X_{t+1}$ given $X_t = x$.  All of our results include
the case where $\XX$ is finite and, in this case, the transition density $q(x,
y)$ should be interpreted as a transition matrix. More generally, the term
``density'' should be understood as a synonym for ``probability mass
function.''

\begin{assumption}
    \label{a:i}
    The function $q$ is continuous and the $\ell$-step transition density
    $q^\ell$ is everywhere positive at some $\ell \in \NN$.
\end{assumption}

Continuity can be ignored when $\XX$ is finite.
Positivity of $q^\ell$ for some $\ell$
means that $\{X_t\}$ is both aperiodic and irreducible,  guaranteeing
uniqueness of the stationary distribution $\pi$ and providing regularity for
asymptotic values such as the the risk-adjusted long-run mean consumption
growth rate.  Assumption~\ref{a:i} is either satisfied directly in our
applications or can be validated after an arbitrarily small perturbation (as is the case
for the learning application---see section~\ref{ss:learning}).

For example, consider the consumption growth specification
\begin{align}
    \ln \left( C_{t+1}/C_{t}\right)
        &=\mu_c + X_t + \sigma_c \, \epsilon_{t+1}
    \label{al:sar0}
        \\
    X_{t+1}
        &=\rho X_t + \sigma \, \eta _{t+1}
    \label{al:sar1}
\end{align}
from section~I.A of \cite{bansal2004risks}.  Here $-1 < \rho < 1$ and $\left\{
\epsilon _t \right\} $ and $\left\{ \eta _t \right\} $ are {\sc iid}
standard normal innovations.\footnote{These calculations can be further extended to the case where consumption
growth is a component of a VAR, as in \cite{hansen2008consumption} or
\cite{bansal_kiku_shaliastovich_yaron:2014}.}  With the state process set to $\{X_t\}$,
we have $q(x, \cdot) = N(\rho x, \sigma^2)$  and assumption~\ref{a:i} is satisfied with $\ell = 1$
whenever $\sigma > 0$.

For this model of consumption dynamics, an analytical expression for $\lrM$
exists, even though the same is not true for the continuation value
$V_t$.  To see this, observe that $C_n / C_0 = \exp
\left(n \mu_c + \sum_{t=1}^n H_t \right)$ where $H_t := X_t + \sigma_c
\epsilon_t$, so the risk-adjusted long-run mean consumption growth rate is
\begin{equation*}
    \lrM
    =
    \lim_{n \to \infty}
    \left\{
        \rR \exp \left( n \mu_c + \sum_{t=1}^n H_t \right)
    \right\}^{1/n}.
\end{equation*}
Here $\rR$ is the unconditional Kreps--Porteus certainty equivalent operator
defined by $\rR[Y] := \EE[Y^{1-\gamma}]^{1/(1-\gamma)}$.
Noting that $\{H_t\}$ is Gaussian and setting $s_n$ equal to the variance of $\sum_{t=1}^n H_t$,
we have
\begin{equation}
    \label{eq:re}
    \rR \exp \left( n \mu_c + \sum_{t=1}^n H_t \right)
    = \exp \left\{
         n \mu_c + \frac{ (1-\gamma) s^2_n}{2}
        \right\}.
\end{equation}
Using the fact that $\sum_{t=1}^n H_t$ is the sum of the independent terms
$\sum_{t=1}^n X_t$ and $\sigma_c \sum_{t=1}^n \epsilon_t$, along with the AR(1)
dynamics in \eqref{al:sar1}, straightforward calculations lead to
\begin{equation*}
    s^2_n
    = n \sigma_c
    +
    \frac{\sigma^2}{1-\rho^2}
    \left\{
        n +
         2(n-1) \frac{\rho}{1-\rho}
            -  2 \rho^2 \frac{1- \rho^{n-1}}{(1-\rho)^2}
    \right\}.
\end{equation*}
Raising the right hand side of \eqref{eq:re} to the power of $1/n$ and
taking the limit yields
\begin{equation}
    \label{eq:lrma}
    \lrM
    =
    \exp \left\{
        \mu_c + \frac{1}{2}(1-\gamma)
        \left( \sigma_c^2 + \frac{\sigma^2}{(1-\rho)^2} \right)
        \right\}.
\end{equation}
While the unconditional variance of the one-period consumption growth rate is
$\sigma_c + \sigma^2/(1-\rho^2)$, persistence in consumption growth implies that the
long-run variance of the consumption growth rate is given by
$\sigma_c + \sigma^2/(1-\rho)^2$. Typical calibrations of risk aversion set $\gamma > 1$
and higher persistence and volatility then contribute negatively to $\lrM$.
For example, under the parameterization of the consumption process in table~I
of \cite{bansal2004risks},
the expression in \eqref{eq:lrma} evaluates to $1.00045$ when $\gamma=7.5$ and
$0.99964$ when $\gamma=12.5$.\footnote{\label{fn:mp}The parameter values in question are $\mu_c = 0.0015$, $\rho=0.979$, $\sigma=0.00034$ and $\sigma_c = 0.0078$.}


\subsection{The Finite State Case}

\label{ss:fsc}

The finite state setting gives us an alternative view on the risk-adjusted
long-run mean consumption growth rate $\lrM$ and an alternative way to
compute it.  In particular,
if $\XX$ is finite with typical elements $x, y$ and $K$ is the matrix with
$(x, y)$-th element
\begin{equation}
    \label{eq:kij}
    K(x, y) :=
    \sum_{y \in \XX} \int \exp[ (1-\gamma) \kappa(x, y, \epsilon) ] \nu(\diff \epsilon) q(x, y) ,
\end{equation}
then \eqref{eq:deftheta} holds; that is,
\begin{equation}
    \label{eq:lrfc}
    \lrM = r(K)^{1/(1-\gamma)}
    \quad \text{when} \quad
    r(K) = \max_{\lambda \in E} |\lambda|.
\end{equation}
Here $E$ is the set of eigenvalues of $K$, so $r(K)$ is the spectral radius of
the matrix $K$.

To gain some understanding as to why the alternative representation of $\lrM$
in \eqref{eq:lrfc} is valid, note that, for all $x$ in $\XX$ and all $n$ in $\NN$, we have
\begin{equation}
    \label{eq:gsi}
    K^n \1(x)
    = \EE_x
    \exp \left\{
        (1-\gamma) \sum_{t=1}^n \kappa(X_{t-1}, X_t, \epsilon_t)
    \right\}.
\end{equation}
Here $K^n$ is the $n$-th power of $K$ in \eqref{eq:kij} and $\1$ is a vector
of ones.  The expectation $\EE_x$ conditions on $X_0 = x$.
That the identity in \eqref{eq:gsi} is true at $n=1$ follows immediately from
the definition of $K$ in \eqref{eq:kij}, and the case of
general $n$ can be confirmed by induction.
Let $\| \cdot \|$ be the $L_1$ vector norm defined by $\| h \| = \sum_{x \in
\XX} |h(x)| \pi(x)$.  By \eqref{eq:gsi} and the Law of Iterated Expectations, we have
\begin{equation}
    \| K^n \1 \|
    = \sum_{x \in \XX} K^n \1(x) \pi(x)
    = \EE
    \exp \left\{
        (1-\gamma) \sum_{t=1}^n \kappa(X_{t-1}, X_t, \epsilon_t)
    \right\}.
\end{equation}
In other words,
\begin{equation}
    \label{eq:vne}
    \| K^n \1 \|
    = \EE
    \,
    \left( \frac{C_n}{C_0} \right)^{1-\gamma}
    =
    \left[ \rR \left( \frac{C_n}{C_0} \right) \right]^{1-\gamma}.
\end{equation}

Gelfand's formula for the spectral radius tells us that $\| K^n \|^{1/n} \to r(K)$ as $n \to \infty$
whenever $\| \cdot \|$ is a matrix norm, and this result can, in the present
context, be modified to
show that $\| K^n \1 \|^{1/n} \to r(K)$ also holds.\footnote{This is a finite
dimensional version of the local spectral radius result discussed in the
introduction.  A complete proof that
extends to infinite state settings is given in theorem~\ref{t:lsr} of the appendix.}
Connecting the last result with \eqref{eq:vne} gives
\begin{equation}
    \label{eq:piv}
    \lrM^{1-\gamma}
    = \lim_{n  \to \infty}
    \left[ \rR \left( \frac{C_n}{C_0} \right) \right]^\frac{1-\gamma}{n}
    = \lim_{n  \to \infty}
         \| K^n \1 \|^{1/n}
    = r(K),
\end{equation}
which justifies the claim in \eqref{eq:lrfc}.  In several applications below
we use this identity by calculating $r(K)$ numerically and then recovering
$\lrM$ via $\lrM = r(K)^{1/(1-\gamma)}$.

\subsection{Discretization}

\label{ss:dis}

As discussed in the previous section, if the state space is finite, then we
can use the identity $\lrM = r(K)^{1/(1-\gamma)}$ obtained in \eqref{eq:piv}
to calculate $\lrM$, which in turn allows us to compute the stability exponent
$\Lambda$.  If, on the other hand, $\XX$ is not finite, then one option is to discretize the
model, leading to a finite state space and a finite set of transition
probabilities $q(x, y)$, and then proceed as above.  Here we investigate the
accuracy of this procedure.

Our experiment is based on the Bansal--Yaron consumption growth dynamics in
\eqref{al:sar0}--\eqref{al:sar1}, where an analytical expression for $\lrM$ was obtained in
\eqref{eq:lrma}.  We first discretize the Gaussian AR(1) state process
\eqref{al:sar1} using Rouwenhorst's method (\citeauthor{rouwenhorst1995},
\citeyear{rouwenhorst1995}). Then we compute the matrix $K$ in \eqref{eq:kij}
corresponding to this discretized state process, calculate the spectral radius
$r(K)$ using linear algebra routines and, from there, compute the associated
value for $\lrM$ as $r(K)^{1/(1-\gamma)}$. Finally, we compare the result with
the true value of $\lrM$ obtained from the analytical expression
\eqref{eq:lrma}.

Table~\ref{tab:dmp} shows such a comparison at a range of parameter values and
levels of discretization.  The preference parameter $\gamma$ is varied across
the rows, as shown in the first column, while the remaining parameters are
sourced from table~I of \cite{bansal2004risks}, as in footnote~\ref{fn:mp}.
The second column shows the true value of
$\lrM$ for the nondiscretized model.  The remaining columns show the value
computed using the numerical procedure discussed in the previous paragraph at
different levels of discretization.  For example, $D=5$ means that the AR(1)
process \eqref{al:sar1} was discretized into a 5 state Markov chain using the
Rouwenhorst method.  The results show that the discretization based method is
accurate for this model, even for relatively coase approximations.

{\small
\begin{table}
    \centering
    \begin{tabular}{l|c|cccc}
    \toprule
     &  true value &  $D=5$  & $D=50$   &  $D=100$     &     $D=200$ \\
    \midrule
    $\gamma=7.5$ & 1.0004504 & 1.0004998 & 1.0004549 & 1.0004527 & 1.0004516 \\
    $\gamma=10.0$ & 1.0000466 & 1.0001658 & 1.0000584 & 1.0000525 & 1.0000496 \\
    $\gamma=12.5$ & 0.9996430 & 0.9998662 & 0.9996673 & 0.9996552 & 0.9996491 \\
    \bottomrule
    \end{tabular}
    \vspace{0.8em}
    \caption{\label{tab:dmp}True value and discrete approximation of $\lrM$, Bansal--Yaron model}
\end{table}
}

\subsection{A Monte Carlo Method}

\label{ss:mcm}

One potential issue with the discretization based method discussed in
section~\ref{ss:dis} is that the algorithm is computationally inefficient
when the state space is large.  For this reason we also propose a Monte Carlo
method to calculate an approximation to $\lrM $ that is less susceptible to
the curse of dimensionality.  The first step is to replace the limit in the
definition of $\lrM$ by some finite but large $n$, which leads to
\begin{equation}
    \label{eq:drcn}
    \lrM(n)
    := \left\{
        \EE \left( \frac{C_n}{C_0} \right)^{1-\gamma}
    \right\}^{\frac{1}{1-\gamma} \frac{1}{n} }
\end{equation}

Next, the expectation in \eqref{eq:drcn} is replaced by a sample
average over $m$ independent consumption paths, generated according to the
specifications of the model in question.   In particular, with $\{C_t^{(j)}\}$
as the $j$-th of the $m$ consumption paths,\footnote{Since $\EE$ in
\eqref{eq:drcn} is the unconditional expectation, we draw the initial state
$X^{(j)}_0$ from its stationary distribution when computing $\{C_t^{(j)}\}$.}
we take the approximation
\begin{equation}
    \label{eq:drcmn}
    \lrM(m, n)
    := \left\{
        \frac{1}{m} \sum_{j=1}^m \left( \frac{C_n^{(j)}}{C_0^{(j)}} \right)^{1-\gamma}
    \right\}^{\frac{1}{1-\gamma} \frac{1}{n} }
\end{equation}

With fixed $n$ and $Y^{(j)} := ( C_n^{(j)}/C_0^{(j)})^{1-\gamma}$,
the Strong Law of Large Numbers
yields
$\frac{1}{m} \sum_{j=1}^m Y^{(j)}\to \EE \, Y $ as $m \to \infty$ with
probability one.  However, this result is only asymptotic and our main concern is
whether the estimator $\lrM(m, n)$ has good properties when $m$ and $n$ are
moderate.  To test this, we again use the lognormal consumption model and
compare our approximations of $\lrM$ with the true value obtained via the
analytical expression given in \eqref{eq:lrma}.


Table~\ref{tab:t1} illustrates our results.  The consumption path parameters are again chosen to
match \cite{bansal2004risks}, as in footnote~\ref{fn:mp}. The parameter
$\gamma$ is set to 7.5, which matches the first row of table~\ref{tab:dmp}.
The true value of $\lrM$ calculated from the analytical expression
\eqref{eq:lrma} is $1.0004504$, as shown in the caption for the table.  The
interpretation of $n$ and $m$ in the table is consistent with the right hand
side of \eqref{eq:drcmn}.  For each $n, m$ pair, we compute $\lrM(n, m)$ a
total of 1,000 times using independent draws, and then present the mean and
the standard deviation of the sample in the corresponding table cell.
The Monte Carlo approximation is accurate up to three decimal places in
all simulations. Standard deviations are small and decline with $m$.\footnote{One additional
advantage of the Monte Carlo method centered on \eqref{eq:drcmn} is that
simulation of the independent consumption processes can be 
parallelized.  This leads to speed gains approaching two
orders of magnitude in some of our implementations.}

{\small
\begin{table}
    \centering
    \begin{tabular}{l|rrrrr}
    \toprule
    &  $m = 1000$ & $m = 2000$ & $m = 3000$ & $m = 4000$ & $m = 5000$ \\
    \midrule
    $n = 250$ & 1.0006940 & 1.0006905 & 1.0006940 & 1.0006932 & 1.0006934 \\
         & (0.000076) & (0.000048) & (0.000032) & (0.000037) & (0.000029) \\
    $n = 500$ & 1.0006208 & 1.0006091 & 1.0005813 & 1.0005733 & 1.0005775 \\
             & (0.000084) & (0.000066) & (0.000068) & (0.000084) & (0.000062) \\
    $n = 750$ & 1.0005979 & 1.0005762 & 1.0005764 & 1.0005611 & 1.0005523 \\
             & (0.000112) & (0.000096) & (0.000076) & (0.000074) & (0.000092) \\
    \bottomrule
    \end{tabular}
    \vspace{0.6em}
    \caption{Realizations of $\lrM(m, n)$ when $\lrM = 1.0004504$. Values
    shown are mean and standard deviation over 1,000 independent draws.}
    \label{tab:t1}
\end{table}
}

High accuracy in estimating $\lrM$ translates into similarly
high accuracy in estimating $\Lambda = \beta
\lrM^{1-1/\psi}$ once we
introduce the additional preference parameters $\beta$ and $\psi$.
Table~\ref{tab:t1l} illustrates this point.  Each $(m, n)$ cell in the table
gives the mean and standard deviation of 1,000 draws of
\begin{equation}
    \label{eq:drlm}
    \Lambda(m, n) := \beta \lrM(m, n)^{1 - 1/\psi}.
\end{equation}
The draws of $\lrM(m, n)$ are computed as in table~\ref{tab:t1}, while the
remaining parameters are set to $\beta = 0.998$ and $\psi = 1.5$, as in
\cite{bansal2004risks}.   The true value $\Lambda = 0.9981498$ from the caption of
table~\ref{tab:t1l} is calculated as $\Lambda = \beta \lrM^{1 - 1/\psi}$ where
$\lrM$ is obtained from the analytical expression \eqref{eq:lrma}.  The
significance of the actual values is discussed below.

{\small
\begin{table}
    \centering
    \begin{tabular}{l|rrrrr}
    \toprule
    &  $m = 1000$ & $m = 2000$ & $m = 3000$ & $m = 4000$ & $m = 5000$ \\
    \midrule
    $n = 250$ & 0.9982308 & 0.9982297 & 0.9982308 & 0.9982305 & 0.9982306 \\
             & (0.000025) & (0.000016) & (0.000011) & (0.000012) & (0.000010) \\
    $n = 500$ & 0.9982065 & 0.9982026 & 0.9981934 & 0.9981907 & 0.9981921 \\
             & (0.000028) & (0.000022) & (0.000022) & (0.000028) & (0.000021) \\
    $n = 750$ & 0.9981989 & 0.9981916 & 0.9981917 & 0.9981866 & 0.9981837 \\
             & (0.000037) & (0.000032) & (0.000025) & (0.000025) & (0.000031) \\
    \bottomrule
    \end{tabular}
    \vspace{0.6em}
    \caption{Realizations of $\Lambda(m, n)$ when $\Lambda = 0.9981498$}
    \label{tab:t1l}
\end{table}
}

\section{Existence and Uniqueness of Recursive Utilities}

\label{s:ez}

We now state our main theoretical results. For this section we restrict
attention to the case where the state space is compact.  This also covers the
finite state case---and hence any numerical representation of consumption
dynamics.  For the sake of brevity, most of our exposition focuses on the
continuous state case.  Translations to the finite state setting are
straightforward.

Our interest centers on existence, uniqueness and computability of $V_t / C_t$
in \eqref{eq:agg}--\eqref{eq:ce}, although it turns out to be convenient to
solve first for
\begin{equation}
    \label{eq:defg}
    G_t := \left( \frac{V_t}{C_t} \right)^{1-\gamma}.
\end{equation}
We use \eqref{eq:ce},
\eqref{eq:kappa} and \eqref{eq:defg} to rewrite the preference recursion \eqref{eq:agg} as
\begin{equation}
    \label{eq:seq2}
    G_t = \left\{
        1 - \beta + \beta
        \left\{
        \EE_t \, G_{t+1}
        \exp \left[
                (1-\gamma) \kappa(X_t, X_{t+1}, \epsilon_{t+1})
            \right]
        \right\}^{1/\theta}
        \right\}^\theta
\end{equation}
where
\begin{equation*}
    \theta := \frac{1-\gamma}{1-1/\psi}.
\end{equation*}
In terms of a stationary Markov solution $G_t = g(X_t)$, the restriction in
\eqref{eq:seq2} translates to
\begin{equation}
    \label{eq:fpe0}
    g(x) = \left\{
        1 - \beta + \beta
        \left\{
            \int g(y) \int \exp \left[ (1-\gamma) \kappa(x, y, \epsilon)
                \right]
            \nu(\diff \epsilon) q(x, y) \diff y
        \right\}^{1/\theta}
    \right\}^\theta
\end{equation}
for all $x \in \XX$, where $\nu$ is the distribution of
$\epsilon_{t+1}$.\footnote{\label{fn:wc}These transformations utilize homotheticity of
    Epstein--Zin utility. In particular, existence and uniqueness of $V_t/C_t$
    is equivalent to existence and uniqueness of the wealth-consumption ratio,
    which is equal to $(1 - \beta)^{-1} G_t^{1/\theta}$ in this model.}

We solve \eqref{eq:fpe0} for the unknown function $g$ by converting it into a
fixed point problem.  As a first step, we define $Kg$ by
\begin{equation}
    \label{eq:defk}
    K g(x)
    = \int g(y)
    \int \exp[ (1-\gamma) \kappa(x, y, \epsilon) ] \nu(\diff \epsilon)
    q(x, y) \diff y
\end{equation}
The action of the linear operator $K$ on $g$ generalizes the idea of applying the
matrix $K$ in \eqref{eq:kij} to a column vector, as per our discussion of the
finite state case in section~\ref{ss:fsc}.  Now let $\phi$ be the scalar function
\begin{equation}
    \label{eq:defphi}
    \phi(t) = \left\{ 1 - \beta +  \beta \, t^{1/\theta} \right\}^\theta
\end{equation}
on $\RR_+$.\footnote{We adopt the convention $\infty^\alpha =
0$ whenever $\alpha < 0$ so that, in particular, $\phi(0)=0$ when $\theta <
0$.}  Then we define $A$ as the operator mapping $g$ into $Ag$ where
\begin{equation}
    \label{eq:defa}
    A g(x) = \phi( Kg(x)).
\end{equation}
Now \eqref{eq:fpe0} can be written as $g(x) = Ag(x)$. Thus, fixed points
of $A$ coincide with solutions to the recursive utility problem.

\begin{assumption}
    \label{a:c}
    The state space $\XX$ is compact.
\end{assumption}

Assumption~\ref{a:c} includes the case where $\XX$ is finite.  It
holds for any numerical application but is not always satisfied in the
theoretical models we consider.  A treatment
of the unbounded case is provided in section~\ref{s:ub}.

Let $\cC$ be the set of positive functions $g$ on $\XX$ such that $g(X_t)$ has
finite first moment.  With $\Lambda$ as defined in \eqref{eq:kc}, we can state
our main findings:

\begin{theorem}
    [Existence and Uniqueness]
    \label{t:bkl1c}
    If assumptions~\ref{a:i} and \ref{a:c} hold, then $\Lambda$ is well
    defined and the following statements are equivalent:
    \begin{enumerate}
        \item $\Lambda < 1$.
        \item $A$ has a fixed point in $\cC$.
        \item There exists a $g \in \cC$ such that $\{A^n g\}_{n \geq 1}$
            converges to an element of $\cC$.\footnote{Here and below,
                convergence is in terms of absolute mean error (i.e., $L_1$
                deviation).  Thus, the statement that $g_n \to g$ in $\cC$ means that $\int |g_n(x) - g(x)|
            \pi(x) \diff x \to 0$ as $n \to \infty$.}
        \item $A$ has a unique fixed point in $\cC$.
        \item $A$ has a unique fixed point $g^*$ in $\cC$ and $A^n g \to
            g^*$ as $n \to \infty$ for any $g$ in $\cC$.
    \end{enumerate}
\end{theorem}

From a practical perspective, the most useful result in
theorem~\ref{t:bkl1c} is that a unique solution exists precisely when $\Lambda
< 1$.  Another interesting result is the logical equivalence of
(b) and (d), which tells us that at most one solution exists for every parameterization.
A third is that, since (c) is equivalent to
(e), convergence of successive approximations from any starting point in
$\cC$ implies that a unique solution exists and this solution is equal to
the limit of the successive approximations from every initial condition.
Thus, if computing the solution to the model at a given set of parameters is
the primary objective, then convergence of the iterative method itself
justifies the claim that the limit is a solution, and that no other solution exists
in $\cC$.

Condition (a) in theorem~\ref{t:bkl1c}, which translates to
$\beta \, \lrM^{1-1/\psi} < 1$, separates the
contributions of impatience, intratemporal risk adjustment and intertemporal substitutability
of consumption to the valuation of the consumption stream.
More impatience (lower $\beta $) reduces
$\Lambda$, promoting finite valuation. Higher intratemporal risk adjustment $\gamma$ decreases
$\lrM$ but its impact on $\Lambda$ depends on the
elasticity of intertemporal substitution. When preferences are elastic, with $\psi >1$, the income effect of
a higher value of future consumption arising from an increase in
$\lrM $ is stronger than the change in the marginal rate
of substitution between current and future consumption. Since $V_t/C_t$ is
denominated in utils per unit of current consumption, it
diverges to infinity as $\lrM $ increases. The
relative strength of the income and substitution effect switches when $\psi
<1$, and $\Lambda<1$ is violated when $\lrM$ is sufficiently small.

Note that, since the test value $\Lambda$ only depends on consumption through
the risk-adjusted long-run mean consumption growth rate $\lrM$, transitory
details of the consumption process are irrelevant for existence of the
continuation value. This insight is used in the applications below,
particularly when we compare models with and without learning
(sections~\ref{ss:msd} and \ref{ss:learning} respectively).

The condition $\Lambda < 1$ imposes a bound on the average growth rate in the
value of long-dated consumption strips as their maturity increases. In an {\sc iid} growth
setting, $\log\Lambda$ is exactly equal to this (constant) growth rate, which
must be negative for a finite wealth-consumption ratio, and hence the
continuation value, to exist.\footnote{Section~\ref{oapp-s:iid_growth} in the
    online appendix outlines these calculations.}  In other words, the risk
adjusted growth rate of consumption $\lrM$ has to be sufficiently low relative
to the risk-free rate.  In more general settings, $\Lambda$ is tied to the
principal eigenvalue of a valuation operator used in the pricing of
consumption strips.  See
\cite{borovicka_hansen_scheinkman:2016} and \cite{qin2017long} for more
details on these operators.

The sufficiency component of theorem~\ref{t:bkl1c} uses fixed point theory for
monotone operators, applied to the operator $A$ and focusing on specific
existence and uniqueness results originally due to \cite{du1990fixed}.  In the
introduction we discussed these results briefly, where discussion was
limited to the monotone concave case.  In fact the results in \cite{du1990fixed} can
handle both monotone concave and monotone convex operators.  The operator $A$
always falls into one of these categories, being concave when $\theta < 0$ or
$\theta \geq 1$ and convex otherwise---as follows from the convexity and
concavity properties of the function $\phi$ defined in \eqref{eq:defphi},
combined with the linearity of $K$.


The proof also uses spectral results based around the Krein--Rutman
theorem, which extends Perron--Frobenius theory to function spaces.  In
particular,  both the sufficiency and the necessity arguments exploit the
spectral radius identity \eqref{eq:deftheta} from the introduction,  which
relies on an extension of a local spectral radius result originally due to V.\
Ya.\ Stet'senko.  See theorem~\ref{t:lsr} in the appendix for the original
result, theorem~\ref{t:lsr2} for the extension and proposition~\ref{p:lsr} for
the proof of \eqref{eq:deftheta}.

\section{Applications}

\label{s:a}

Next we demonstrate how theorem~\ref{t:bkl1c} can
be used to obtain existence and uniqueness of solutions in applied settings.
We begin with relatively simple applications and then progress to more
realistic specifications of consumption streams.

\subsection{Consumption With a Deterministic Time Trend}

\label{ss:aj}

Consider a model where consumption obeys the geometric trend specification
\begin{equation}
    \label{eq:ac}
    C_t = \tau^t X_t,
\end{equation}
where $\tau$ is a positive scalar and $\{X_t\}$ is positive, compactly
supported and {\sc iid}, as in, say, \cite{alvarez2005using}.  Then
the risk-adjusted long-run mean consumption growth rate is
\begin{equation}
    \label{al:a1}
    \lrM
    = \lim_{n \to \infty }
            \left\{
                \rR
                \left(
                    \frac{\tau^n X_n}{X_0}
                \right)
                \right\}^{\frac{1}{n}}
    = \tau
        \lim_{n \to \infty }
            \left\{
             \EE \left( \frac{X_n}{X_0} \right)^{1-\gamma}
         \right\}^{1/(n (1-\gamma))}
    = \tau,
\end{equation}
where the last equation follows from the fact that
the expectation is constant (since the state process is {\sc iid}).
Appealing to theorem~\ref{t:bkl1c}, we find the exact necessary and sufficient
condition for a unique solution to exist is
\begin{equation}
    \label{eq:lddt}
    \Lambda = \beta \, \tau^{1-\frac{1}{\psi }} < 1.
\end{equation}

Since the consumption process \eqref{eq:ac} is trend-stationary, its
stochastic component is immaterial for the long-run distribution of
consumption growth. This is why the risk aversion parameter $\gamma$ does not
enter into the expression for $\Lambda$ obtained in
\eqref{eq:lddt}.\footnote{The absence of the risk aversion parameter in
    $\Lambda$ contrasts with an earlier condition provided by
    \cite{alvarez2005using}, which does depend on $\gamma$.  However, our
    condition $\Lambda < 1$ is both necessary and sufficient, from which we
can conclude that the risk aversion parameter is indeed irrelevant to
existence and uniqueness of a solution.  See section~\ref{ss:ca} for more
discussion.}




\subsection{Markov Switching Dynamics}

\label{ss:msd}

To illustrate the finite state case, consider the two state Markov switching
specification for consumption growth of \cite{johannes2016learning}, with
\begin{equation}
    \label{eq:kappaj}
    \ln (C_{t+1}/ C_t) = \mu(X_{t+1}) + \sigma(X_{t+1}) \, \epsilon_{t+1}
\end{equation}
where $\XX = \{1, 2\}$ and baseline parameter values as in table~\ref{tab:sometab}.
The process $\{\epsilon_t\}$ is {\sc iid} and standard normal.
We set aside for now the learning component of
\cite{johannes2016learning} and assume that $X_t$ is fully observable.
(The learning problem is treated in section~\ref{ss:learning}.)
The conditions of assumption~\ref{a:i} are satisfied with $\ell=1$.

\begin{table}
  \centering
  \begin{tabular}{l|ccccccc}
      \toprule
      parameter & $\mu(1)$ & $\mu(2)$ & $\sigma(1)$ & $\sigma(2)$ & $q(1, 1)$ & $q(2, 2)$
               & $\gamma$
      \\
      \midrule
      value & 0.007 & 0.0013 & 0.0015 & 0.0063 & 0.93 & 0.83 & 10.0 \\
      \bottomrule
  \end{tabular}
      \vspace{1em}
  \caption{Parameter values for the Markov switching model} \label{tab:sometab}
\end{table}

Our first step is to compute $\lrM$ for this model using
the spectral radius method dicussed in section~\ref{ss:fsc}.
We take $K$ from \eqref{eq:kij}, which in the
present setting reduces to the $2 \times 2$ matrix
\begin{equation}
    \label{eq:kij2}
    K(i, j) = K_{ij} =
    \exp
        \left[
            (1-\gamma) \mu(j) + \frac{1}{2} (1-\gamma)^2 \sigma(j)^2
        \right] q(i, j) .
\end{equation}
After inserting the parameters from table~\ref{tab:sometab},
we calculate its eigenvalues, evaluate $r(K)$ as the maximal eigenvalue in
modulus and then apply $\lrM = r(K)^{1/(1-\gamma)}$ from \eqref{eq:piv}.  The
result is $\lrM = 1.005$.  Figure~\ref{f:gms} shows how $\lrM$ varies as
$\gamma$ and $\sigma(1)$ shift around their baseline values.  Both decrease
the risk-adjusted long-run mean consumption growth rate as they rise, with the negative
effect of $\gamma$ intensifying as volatility in the good state grows.

\begin{figure}
    \centering
    \hspace{-1em}
    \scalebox{0.55}{\includegraphics[clip=true, trim=0mm 0mm 0mm 0mm]{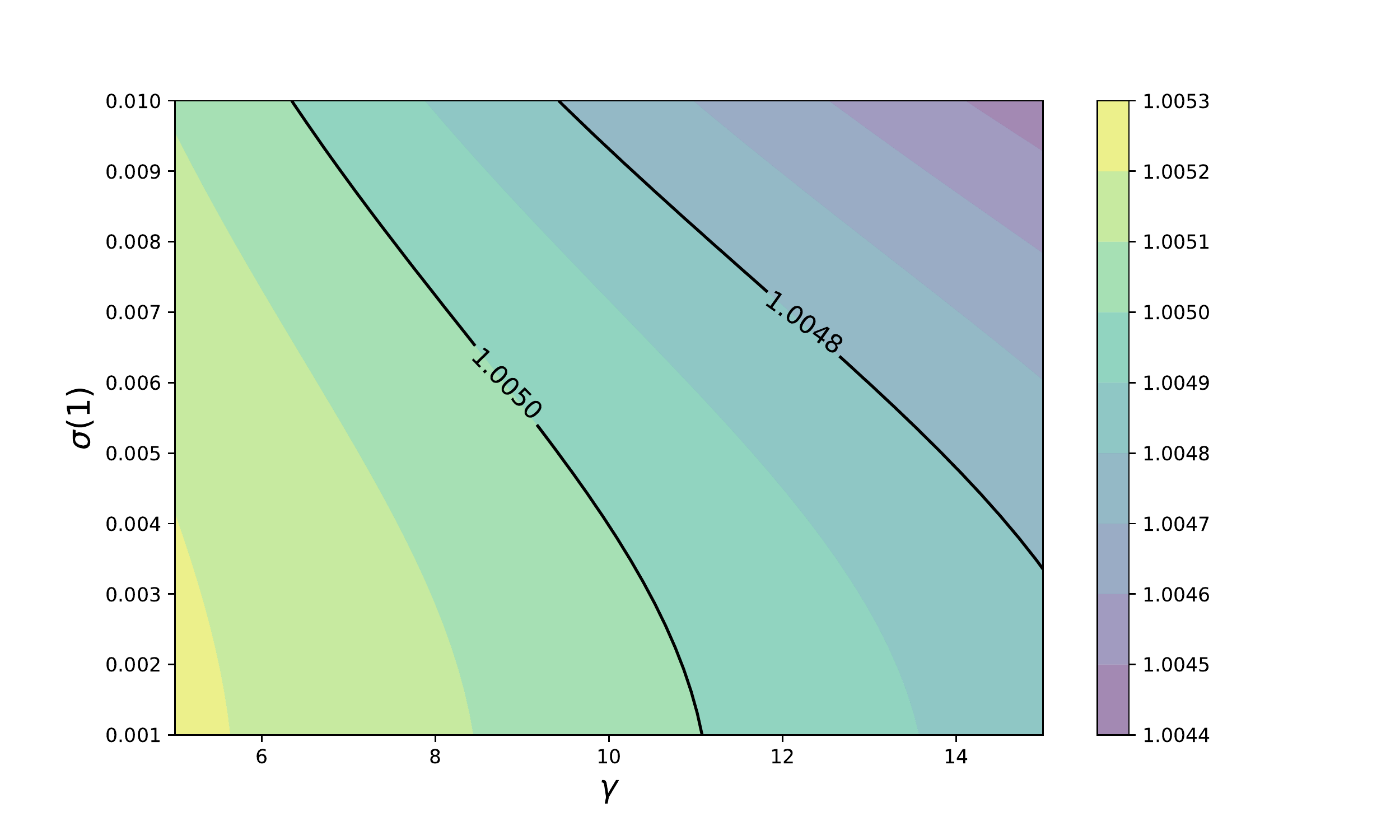}}
    \caption{\label{f:gms} $\lrM$ as a function of $\gamma$ and $\sigma(1)$ in the Markov switching model}
\end{figure}

With $\lrM$ in hand, we can calculate $\Lambda = \beta \lrM^{1-1/\psi}$.
With preference parameters $\psi=1.5$ and
$\beta=0.998$ from \cite{johannes2016learning}, we obtain $\Lambda =
0.99567$.  Hence condition (a) in theorem~\ref{t:bkl1c} is satisfied and
  statements (b)--(e) hold true.  In particular, a unique and globally
  attracting solution to the recursive utility problem exists.
  Conversely, if we (arbitrarily) pair the values $\psi=1.97$ and $\beta=0.999$ from
  \cite{schorfheide2018identifying} with this consumption process, then $\Lambda = 1.00147$ and, by part
  (b) of theorem~\ref{t:bkl1c}, no solution exists.

These results turn out to be significant even for the model that does include
learning, as the next section details.

\subsection{Markov Switching with Learning}

\label{ss:learning}

Consider again the Markov switching dynamics from the previous section
but suppose as in \cite{johannes2016learning} that the investor does not observe the state process $\{X_t\}$. Instead,
she forms a Bayesian posterior belief about $X_t$ based on her time-$0$ prior
and observations of consumption growth up to time $t$. The investor's filtering
problem can be represented recursively, with $\overline{X}_t \in [0,1]$
denoting the probability that $X_t=1$ under the investor's information set.
This new state variable can be shown to follow the law of motion $\overline{X}_{t+1} =
h(\overline{X}_t,Z_{t+1})$, where $Z_{t+1} = \ln(C_{t+1}/C_t)$ is observed
consumption growth and
\begin{equation}
    h(x,z) = \frac{\left[xq(1,1)+(1-x)q(2,1)\right]\varphi_1(z)}{\left[xq(1,1)+(1-x)q(2,1)\right](\varphi_1(z)-\varphi_2(z))+\varphi_2(z)}.
    \label{eq:defh}
\end{equation}
Here $\varphi_j$ is the density of $Z_{t+1}$ conditional on $X_t=j$.
Details of the derivation of this standard filtering problem are
provided in the online appendix.
The matrix $K$ in \eqref{eq:kij2}, which served as the valuation operator in
the Markov switching problem without learning, is replaced by the operator
\begin{equation}
    \overline{K}g(x)
    =  \int g(y) \, [\xi(1)y+\xi(2)(1-y)] \, \overline{q}(x,y) \diff y
    \label{eq:Kbar}
\end{equation}
where
\begin{equation*}
    \xi(i) = \exp\left( (1-\gamma)\mu(i) + \frac12 (1-\gamma)^2 \sigma^2(i)\right)
\end{equation*}
for $i=1, 2$ and $\overline{q}$ represents the transition density arising from the
subjective belief of the Bayesian learner.  In other words, $\overline{q}(x,
\cdot)$ is the distribution of $h(x, Z_{t+1})$ when $h$ is as defined in
\eqref{eq:defh} and $Z_{t+1} = \ln(C_{t+1}/C_t)$.\footnote{While the normal
    densities $\varphi_j(z)$ specified in
    \cite{johannes2016learning} do not directly imply a transition density
    $\overline{q}(x,y)$ that satisfies assumption~\ref{a:i}, we
    construct in the online appendix an arbitrarily small perturbation of
these densities such that assumption~\ref{a:i} does hold.}

We now evaluate the risk-adjusted mean consumption growth rate \eqref{eq:drc}
under the model with learning, which we denote $\olrM $, in order
to determine the test value $\Lambda$. Denote $\overline{\EE}_t$ the
expectations operator under the investor's time-$t$ information set and
$\overline{\EE}$ its unconditional version. Hence
\begin{equation*}
    \olrM = \lim_{n\to\infty}
        \left\{ \overline{\EE}
            \left[ \left( C_n/C_0 \right)^{(1-\gamma)} \right]^\frac{1}{1-\gamma}\right\}^{1/n}.
\end{equation*}
The Law of Iterated Expectations implies that for any $n\in\NN$
\begin{equation*}
    \min_{X_0}\EE_0 \left[ \left( C_n/C_0 \right)^{(1-\gamma)} \right] \leq \overline{\EE}
            \left[ \left( C_n/C_0 \right)^{(1-\gamma)} \right] \leq \max_{X_0}\EE_0 \left[ \left( C_n/C_0 \right)^{(1-\gamma)}\right].
\end{equation*}
Raising this pair of inequalities to power $1/(n(1-\gamma))$ and taking the limit as $n\to\infty$, we obtain
\begin{equation*}
    \lim_{n\to\infty}
    \left\{ \min_{X_0} \EE_0 \left[ \left( C_n/C_0 \right)^{(1-\gamma)} \right]^\frac{1}{1-\gamma}\right\}^{1/n} \leq \olrM \leq
    \left\{ \max_{X_0} \EE_0 \left[ \left( C_n/C_0 \right)^{(1-\gamma)} \right]^\frac{1}{1-\gamma}\right\}^{1/n}.
\end{equation*}
However, the limits of the conditional expectations on the left- and right-hand side of this relationship are independent of
the initial state $X_0$ due to ergodicity of the two-state Markov chain.
Consequently, $\olrM  = \lrM $, and the parameters for which
unique continuation values exist coincide in the full and partial information
models.

This result holds despite the fact that the two models have fundamentally
different state spaces and transition densities. The reason for the result is
the transitory impact of state uncertainty on the conditional distribution of future consumption growth. Our test value $\Lambda$
depends only on the long-run distribution of consumption growth under
the investor's belief. Transitory deviations in the investor's belief about future
consumption growth relative to the data generating process, driven by learning
about the unobserved state, are inconsequential in the long run and therefore
irrelevant for existence of a finite continuation value, despite the fact that
they affect the conditional expectations $\overline{\EE}_t$ in every step of
the continuation value recursion \eqref{eq:fpe0}.

As a consequence, in this class of learning models, existence and
uniqueness of the continuation value can be more easily evaluated by
computing $\Lambda$ in the full information model as in section~\ref{ss:msd}, which has a simpler state
space and transition density.

\subsection{Long-Run Risk}

\label{ss:lrr}

Next consider the consumption specification adopted in \cite{schorfheide2018identifying}, where
\begin{align}
    \ln (C_{t+1} /  C_t) &
    = \mu_c + z_t + \sigma_{c, t} \, \eta_{c, t+1},
    \label{al:ssyc1}
    \\
    z_{t+1} = \rho \, z_t
        & + \sqrt{1 - \rho^2} \, \sigma_{z, t} \, \eta_{z, t+1},
    \label{al:ssyc2}
    \\
    \sigma_{i, t} = \phi_i \, \bar{\sigma} \exp(h_{i, t})
    & \quad \text{with} \quad
    h_{i, t+1} = \rho_{h_i} h_i + \sigma_{h_i} \eta_{h_i, t+1},
    \quad i \in \{c, z\}.
    \label{al:ssyc3}
\end{align}

The innovations $\{\eta_{i, t}\}$ and $\{\eta_{h_i, t}\}$ are {\sc iid} and
standard normal for $i \in \{c, z\}$.  The state vector associated with these
consumption dynamics can be represented as $X_t = (h_{c, t}, h_{z,
t}, z_t)$.  The consumption process parameters used by
\cite{schorfheide2018identifying} are shown in table~\ref{tab:ssy_params}, while
the preference parameters are $\gamma=8.89$, $\beta=0.999$ and $\psi=1.97$.

\begin{table}
  \centering
  \begin{tabular}{ccccccccc}
      $\mu_c$ & $\rho$ & $\phi_z$ & $\bar \sigma$ & $\phi_c$ & $\rho_{h_z}$ & $\sigma_{h_z}^2$ & $\rho_{h_c}$ & $\sigma_{h_c}^2$  \\
      \hline
      0.0016 & 0.987 & 0.215 & 0.0035 & 1.0 & 0.992 & 0.0039 & 0.991 & 0.0096
      \vspace{0.6em}
  \end{tabular}
  \caption{Parameterization of the consumption process with long run risk.
  Parameter values are posterior median estimates from
  \cite{schorfheide2018identifying}, Table VII, estimation with
  consumption and financial markets data. \label{tab:ssy_params}}
\end{table}

All of the conditions of theorem~\ref{t:bkl1c} hold apart from compactness of
the state space, which fails because $X_t = (h_{c, t}, h_{z, t}, z_t)$ can take values in
all of $\RR^3$.    However, the fact that the correlation coefficients $\rho$,
$\rho_c$ and $\rho_z$ are less than one in absolute value means that the state
space can be compactified by truncating the standard normal shocks $\eta_z, \eta_{h_c}$ and $\eta_{h_z}$.  For now this is the path that we
pursue.    In this compactified setting, theorem~\ref{t:bkl1c} applies and a unique and
globally stable solution exists if and only if $\Lambda < 1$.

To evaluate $\Lambda$ at a specific level of truncation, we begin with the Monte Carlo method introduced in
section~\ref{ss:mcm}.  This necessitates the use of a random number generator,
so the compactification is implemented automatically by truncating the
innovations in absolute value to the largest double precision floating point
number, without the need to impose further restrictions on the state
space.\footnote{This evaluates to $2^{53}$, which is $\approx 10^{16}$.  A far
    smaller truncation point would be adequate but there are no obvious
    benefits to implementing such a modification, since any truncation of the
    innovations generates a compact state space, regardless of how large.
    Theorem~\ref{t:bkl1c} then applies.  Further discussion of the effects of
truncation is provided below.} Table~\ref{tab:cmc_mn} shows summary statistics
for 1,000 draws of $\Lambda(m, n)$ as $n$, the time series length for
consumption, is varied
across the rows, while $m$ is held fixed at 1,000.\footnote{See equations \eqref{eq:drcmn} and
\eqref{eq:drlm} for the definition of $\Lambda(m, n)$. Consumption paths are
generated according to \eqref{al:ssyc1}--\eqref{al:ssyc3}, with independent
normal variates supplied by a standard random number generator.} In all cases, the mean
estimate lies close to 0.9994 and the standard deviation in the sample is
small.  Other values of $m$ and $n$ produce similar numbers.
Recalling our earlier
finding that $\Lambda(m, n)$ closely approximates $\Lambda$ for moderate
choices of $m$ and $n$ (see, e.g., table~\ref{tab:t1l}), this leads us to infer with
a high degree of certainty that the valuation problem for this compactified model
 has a unique and globally stable solution.

\begin{table}
    \centering
    {\small
    \begin{tabular}{cccccccc}
        $n$ & mean & std & min & 25\% & 50\% & 75\% & max
        \\
        \hline
        \hline
        500 & 0.999408 & 0.000111 & 0.998564 & 0.999385 & 0.999441 & 0.999473 & 0.999529
        \\
        1000 & 0.999384 & 0.000093 & 0.998551 & 0.999351 & 0.999409 & 0.999446 & 0.999517
        \\
        1500 & 0.999401 & 0.000061 & 0.999024 & 0.999386 & 0.999421 & 0.999441 & 0.999466
        \\
        \hline
      \vspace{0.1em}
    \end{tabular}
    }
    \caption{Descriptive statistics for 1,000 draws of $\Lambda(m, n)$ when $m=5000$}
    \label{tab:cmc_mn}
\end{table}

While the values of $\Lambda(m, n)$ are all close to 1 at the given parameterization of
\cite{schorfheide2018identifying},  the outcome $\Lambda < 1$ survives reasonably large deviations in the parameters, which
further supports the conclusion that the valuation problem is globally stable.
Figure~\ref{f:ssy} illustrates this robustness by presenting the test value
$\Lambda$ at neighboring parameterizations.  Values are computed by Monte
Carlo, with $m=n=1000$.  The neighboring parameterizations are obtained by varying two
parameters while holding others constant.  Figure~\ref{f:psi_mu_ssy_spec_rad}
shows values obtained for $\Lambda$, represented by a contour plot, as $\psi$
and $\mu_c$ are varied.
Figure~\ref{f:beta_psi_ssy_spec_rad} shows the same as
$\beta$ and $\psi$ are varied.
Parameterizations to the southwest of the 1.0 contour line are
globally stable, while those to the northeast yield $\Lambda > 1$ and hence,
by part (b) of theorem~\ref{t:bkl1c}, no solution exists.  For both subfigures, the
\cite{schorfheide2018identifying} parameterization falls well within the
interior of the stable set.

\begin{figure}[htb]
    \begin{subfigure}[a]{0.9\textwidth}
        \centering
        \scalebox{0.5}{\includegraphics[clip=true, trim=0mm 0mm 0mm 0mm]{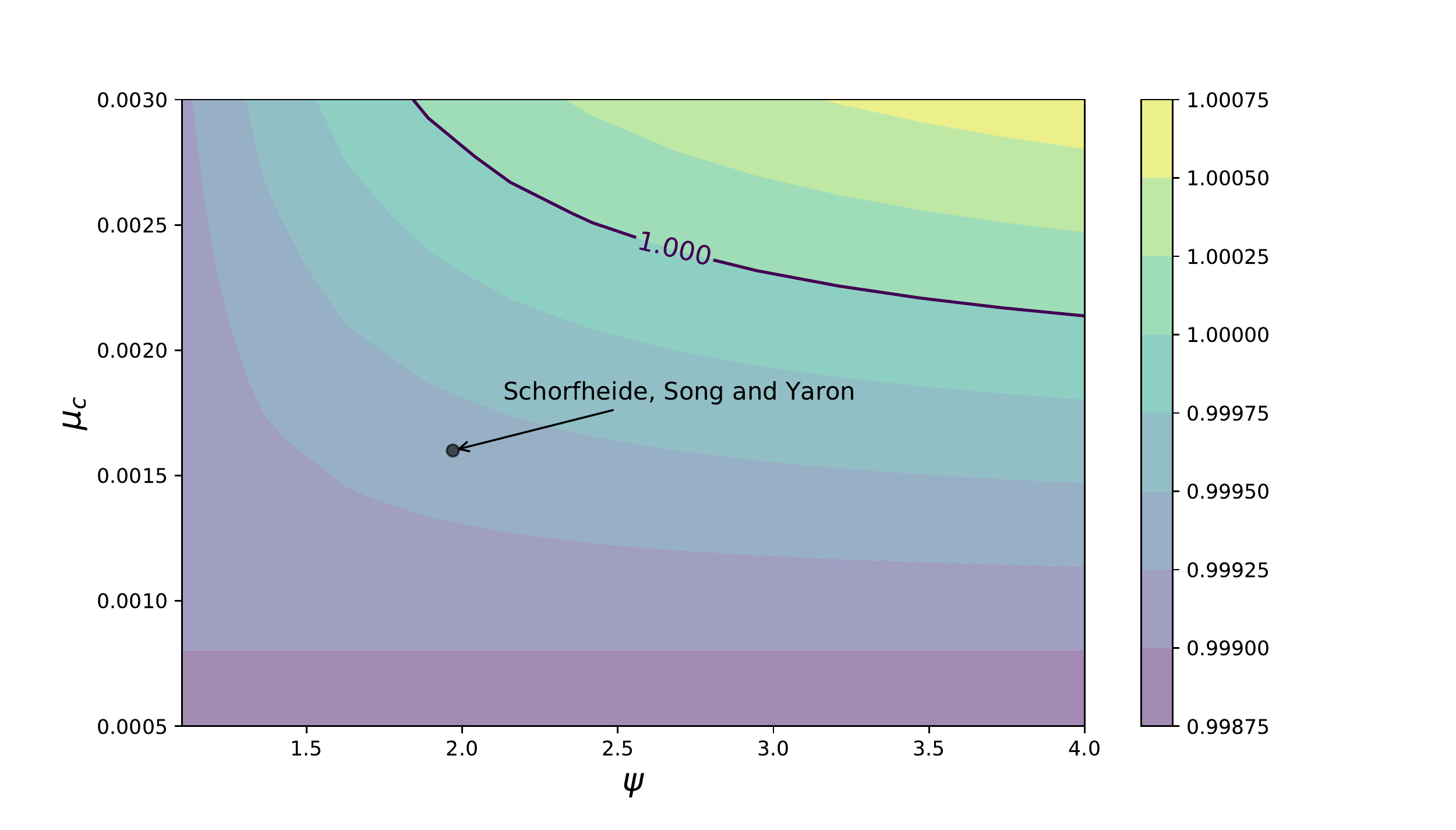}}
        \caption{\label{f:psi_mu_ssy_spec_rad} Changes in test value, $\psi$ vs $\mu_c$}
    \end{subfigure}

    \vspace{0.5em}
    \begin{subfigure}[c]{0.9\textwidth}
        \centering
        \scalebox{0.5}{\includegraphics[clip=true, trim=0mm 0mm 0mm 0mm]{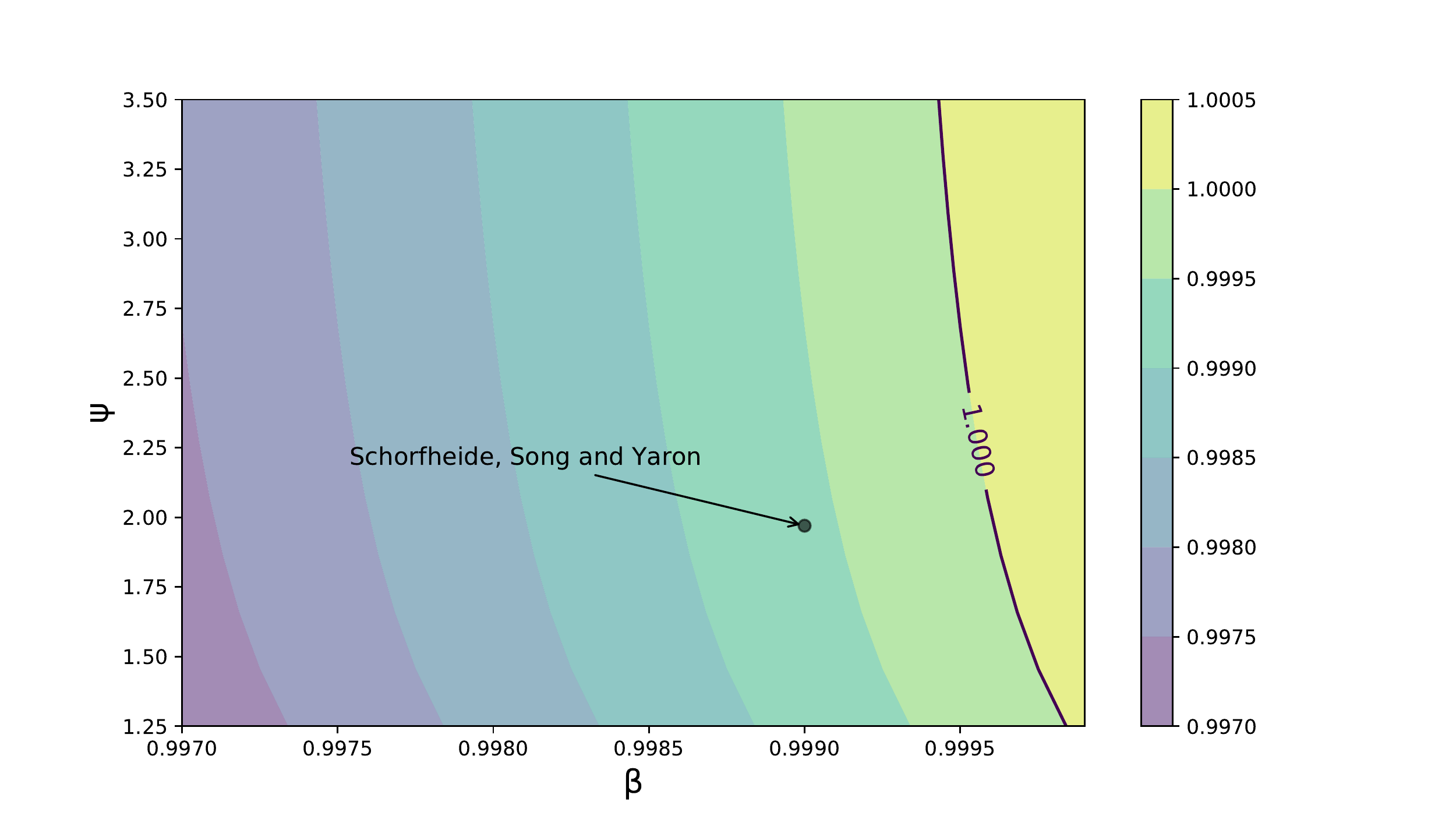}}
        \vspace{0.1em}
        \caption{\label{f:beta_psi_ssy_spec_rad} Changes in test value, $\beta$ vs $\psi$}
    \end{subfigure}
    \vspace{0.5em}
    \caption{\label{f:ssy} A contour map of $\Lambda$ values near the SSY parameterization}
\end{figure}

None of the results presented above speak directly to the existence and
uniqueness in the original theoretical model, where innovations are not
truncated and the state space is unbounded.   There are two independent
questions that need to be considered here.  First, what is the value of
$\Lambda$ at the \cite{schorfheide2018identifying} parameterization when
shocks are not truncated?  Second, what does the value of $\Lambda$ actually
imply for the unbounded state space, given that theorem~\ref{t:bkl1c} is not applicable?

Neither question is fully answered here, although we can make reasonable
conjectures.  Regarding the first question, figure~\ref{f:tmc} studies the
impact of truncation of the innovations in this long-run risk model.  The
truncation value on the horizontal axis is the number of standard deviations
at which the normal shocks in \eqref{al:ssyc2} and \eqref{al:ssyc3} are
truncated.  The values on the vertical axis are the corresponding values of
$\Lambda(m, n)$, estimated using the Monte Carlo method with $n=m=1000$.  The
dashed line shows the value of $\Lambda(m, n)$ when no explicit truncation is
imposed (although we are working with 64bit floating point numbers, so
truncation at around $10^{16}$ standard deviations is implicit).  The first
point made clear by the figure is that the impact of truncation becomes
negligible once the truncation level for the standard normal innovations
reaches 3.  This is consistent with our knowledge of the standard normal
distribution (i.e., rapidly decreasing tails and 99.7\% of probability mass
within three standard deviations of the mean).  The second point is that
relaxing truncation shifts $\Lambda$ down rather than up.\footnote{This is
    consistent with expectations because weakening truncation expands the
    tails of the shocks, increasing volatility. Given the relatively large
    value of $\gamma$ adopted in \cite{schorfheide2018identifying}, higher
    volatility tends to decrease the
    risk-adjusted mean consumption growth rate (see \eqref{eq:lrma} for
    intuition) and hence $\Lambda$.  We tested
        different combinations of $m$ and $n$ in these calculations but none
    altered our conclusions.} Thus, there is strong evidence that
    $\Lambda < 1$ holds for the original unbounded model.

\begin{figure}
    \centering
    \hspace{-1em}
    \scalebox{0.66}{\includegraphics[clip=true, trim=0mm 0mm 0mm 0mm]{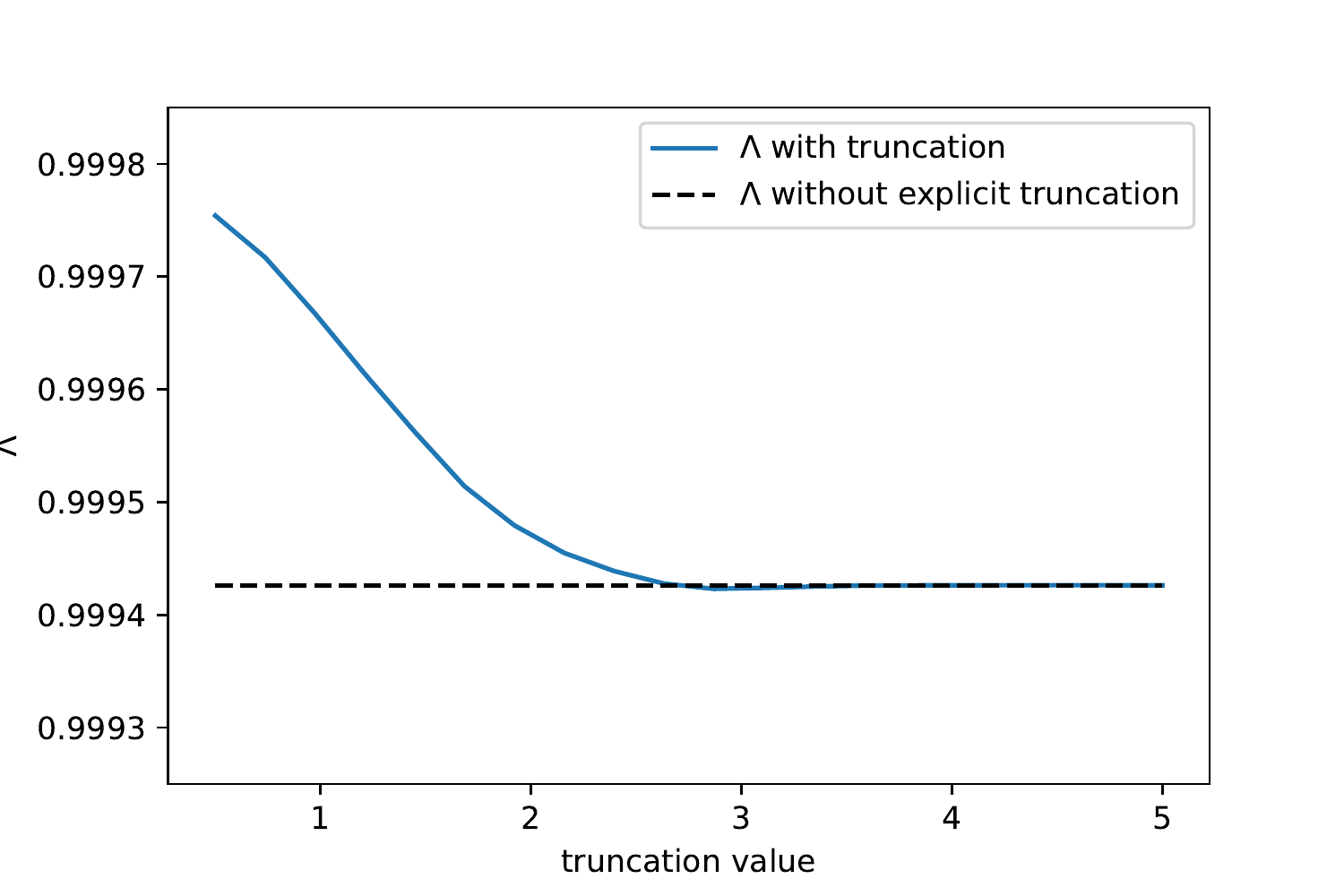}}
    \caption{\label{f:tmc} The impact of truncation on $\Lambda$ in the long run risk model}
\end{figure}

The second question asks what does the condition $\Lambda < 1$ imply when the
state space is not compact and theorem~\ref{t:bkl1c} cannot be applied. Some
guidance in the unbounded setting is provided by theorem~\ref{t:bkgts} from
section~\ref{s:ub} below, which shows that necessity and sufficiency of
$\Lambda < 1$ for existence extends to the unbounded case.   However,
theorem~\ref{t:bkgts} relies on a side condition on $K$ that is difficult to
evaluate in a complex nonlinear model such as the specification of
\cite{schorfheide2018identifying}. At the same time, this condition is
standard in the literature that uses eigendecompositions to analyze long-run
valuations.  Section~\ref{s:ub} gives more details.


\subsection{A Discretized Long-Run Risk Model}

\label{ss:lrrd}

A frequently used solution method is to discretize the long-run risk model to
a finite state space, since such a model is relatively easy to manipulate and
allows for straightforward calculation of all endogenous quantities.
We also use the model to compute the equilibrium wealth-consumption
ratio at a range of parameter values in order to shed light on the asset
pricing consequences of moving in the parameter space toward the region where
a solution ceases to exist.

Discretization is achieved by applying the Rouwenhorst method to the laws of
motion for the state variables (equations \eqref{al:ssyc2} and
\eqref{al:ssyc3}) in each of the three dimensions. The state space is then
finite and the Markov chain for the state process is aperiodic and
irreducible, so assumptions~\ref{a:i} and \ref{a:c} hold. Hence, by
theorem~\ref{t:bkl1c}, a unique and globally stable solution exists for the
discretized model when $\Lambda < 1$.  We found this to be true at the baseline
parameterization for all levels of discretization we considered.\footnote{To handle the stochastic volatility component,
    we proceed as follows:  First we apply the Rouwenhorst method
    independently to $\{h_{c, t}\}$ and $\{h_{z, t}\}$, both of which are
    linear AR(1). Then we translate the results into discretized dynamics for
    $\{\sigma_{c, t}\}$ and $\{\sigma_{z, t}\}$, with $H$ and $I$ possible
    values respectively. Last, for each of the $I$
possible values of $\sigma_z$, we again use the Rouwenhorst method to
discretize $\{z_t\}$ across $J$ possible states.  With $H=I=J=3$, we obtain
$\Lambda = 0.99944$.  For the finer discretizations the value of $\Lambda$ tends
to decline, which is consistent with the intuition we obtain from
figure~\ref{f:tmc}.}

The contour diagrams in figure~\ref{f:ssy} suggest that, if we shift the
parameters of the model sufficiently far from the baseline setting of
\cite{schorfheide2018identifying}, we can find parameterizations where
$\Lambda \geq 1$.  The same is true in the discretized setting.  How should we interpret such outcomes?  One way to understand failure
of existence of a solution to the valuation problem in this model is to study
what happens to equilibrium quantities as we move close to the boundary
between stability and instability.  Consider, for example,
table~\ref{tab:mwc},  which shows how the mean wealth-consumption ratio varies
with the parameters $\psi$ and $\mu_c$, while other parameters are held fixed
at the default values in table~\ref{tab:ssy_params}.
The equilibrium wealth-consumption ratio is equal in this model to $(1 - \beta)^{-1}
g^*(X_t)^{1/\theta}$, where $g^*$ is the fixed point of the operator $A$
discussed in theorem~\ref{t:bkl1c} (see footnote~\ref{fn:wc}).  The function
$g^*$ is computed by iterating with $A$ on the arbitrary initial choice $g
\equiv 1$, a process for which theorem~\ref{t:bkl1c} guarantees convergence
whenever $\Lambda < 1$.\footnote{The calculation was carried out for the discretized representation of the state process,
    which was in turn computed via multiple iterations of the Rouwenhorst
    technique, in the manner discussed above.  In each case we set $H=I=J=3$,
so the state space had 27 elements.  Iteration continued until the maximal
absolute deviation between successive iterates fell below 1e-6.}

Table~\ref{tab:mwc} covers some of the parameter values explored in
figure~\ref{f:psi_mu_ssy_spec_rad}, which showed that no solution
exists when both $\mu_c$ and $\psi$ are sufficiently large.  For those
pairs---or, more precisely, for any pair where $\Lambda \geq 1$---we printed
the string NA to indicate that no solution exists.  For other pairs we printed
the mean wealth-consumption ratio.  Table~\ref{tab:mwc} shows that, as the
parameters approach the boundary between stability and instability, the mean
wealth-consumption ratio explodes.  This is the nature of nonexistence of
solutions: parameters are such that wealth and forward looking valuations are
infinite.

\begin{table}
    \centering
    {\small
    \begin{tabular}{ccccccc}
     &   $\psi$ = 1.1 & $\psi$ = 1.68 & $\psi$ = 2.26 & $\psi$ = 2.84 & $\psi$ = 3.42 & $\psi$ = 4.0 \\
    \hline \hline
    $\mu_c$ = 0.0030 & 1290.3  & 46604.4  & NA  & NA  & NA  & NA  \\
    \hline
    $\mu_c$ = 0.0025 & 1219.3  & 4610.7  & 4.6e+25  & NA  & NA  & NA  \\
    \hline
    $\mu_c$ = 0.0020 & 1155.7  & 2423.3  & 4986.7  & 12840.7  & 3596674.7  & 1.7e+31  \\
    \hline
    $\mu_c$ = 0.0015 & 1098.4  & 1642.7  & 2142.0  & 2600.6  & 3022.8  & 3412.6  \\
    \hline
    $\mu_c$ = 0.0010 & 1046.5  & 1242.0  & 1362.4  & 1443.9  & 1502.7  & 1547.2  \\
    \hline
    $\mu_c$ = 0.0005 & 999.5  & 998.3  & 998.3  & 998.3  & 998.5  & 998.6   \\
    \hline
    \vspace{0.2em}
    \end{tabular}
    }
    \caption{Mean of the wealth-consumption ratio under the stationary
    distribution of the state process, across different $\mu_c, \psi$
    combinations. Monthly consumption in the denominator. As $\psi\to 1$, the
    wealth-consumption ratio converges to the constant $(1-\beta)^{-1}=1000$.}
    \label{tab:mwc}
\end{table}

\subsection{Production Economies}

\label{ss:production}


We have so far assumed that consumption follows an exogenously specified process,
as given in (\ref{eq:kappa}). A large class of models with Epstein--Zin
preferences studied in the macro-finance literature (\cite{tallarini2000risk},
\cite{kaltenbrunner_lochstoer:2010}, \cite{croce:2014}, and many others)
involves endogenously determined consumption processes of the form%
\begin{eqnarray*}
C_{t} &=&c\left( X_{t}\right) A_{t} \\
\log A_{t+1}-\log A_{t} &=&\kappa _{A}\left( X_{t},X_{t+1},\varepsilon
_{t+1}\right) .
\end{eqnarray*}%
For example, suppose we are interested in a production economy where $A_{t}$ is an
exogenous technology process with stochastic growth and are solving for an
endogenous function $c\left( X_{t}\right) $ that represents stationary
deviations from the stochastic trend. Let the assumptions from Section~\ref%
{s:cp} imposed on $X_{t}$ and $\varepsilon _{t}$ hold. Further assume that $%
c\left( x\right) $ is bounded above and away from zero, and denote $%
\overline{\zeta }=\max_{x_{0},x_{n}}\left( c\left( x_{n}\right) /c\left(
x_{0}\right) \right) ^{1-\gamma }$ and $\underline{\zeta }=\min_{x_{0},x_{n}} \left(
c\left( x_{n}\right) /c\left( x_{0}\right) \right) ^{1-\gamma }$. In this
case%
\begin{equation*}
\underline{\zeta }E\left[ \left( \frac{A_{n}}{A_{0}}\right) ^{1-\gamma }%
\right] \leq E\left[ \left( \frac{C_{n}}{C_{0}}\right) ^{1-\gamma }\right]
\leq \overline{\zeta }E\left[ \left( \frac{A_{n}}{A_{0}}\right) ^{1-\gamma }%
\right]
\end{equation*}%
and, as a consequence, the risk-adjusted long-run mean consumption growth
rate $\mathcal{M}_{C}$ from (\ref{eq:drc}) is given by%
\begin{equation}
\mathcal{M}_{C}=\lim_{n\rightarrow \infty } \left[ \rR\left( \frac{A_{n}}{%
A_{0}}\right) \right] ^{1/n}.  \label{eq:mCA}
\end{equation}%
We can therefore infer the region of the parameter space for which the
endogenously determined continuation value exists directly from the
knowledge of the properties of the process $A_{t}$, without knowing the
details of the function $c\left( x\right) $.\footnote{%
\label{f:production}Another sufficient condition for the result in (\ref{eq:mCA}) that is often
satisfied in applications is an exponential decay rate in the correlation
between $c\left( X_{t+n}\right) ^{1-\gamma }$ and $\left(
A_{t+1}/A_{t}\right) ^{1-\gamma }$ as $n\rightarrow \infty $. The argument
can also be further extended to cases when the economy involves additional
endogenous nonstationary state variables cointegrated with $A_{t}$, such as
aggregate capital. In these cases, the transition law for $X_{t}$ may not
satisfy the regularity conditions imposed in Section~\ref{s:cp} but such
economies can often be regularized using perturbation arguments similar to
that described in Section~\ref{oapp-s:learning} of the online appendix.}
This is another manifestation of the fact that our condition for existence
and uniqueness depends only on long-run properties of the consumption
process and transitory details of that process are irrelevant.

\section{Comparisons with Alternative Conditions}

\label{s:pob}

A number of related tests for existence and uniqueness of recursive utilities
were discussed in the introduction.  In this section we provide a brief
comparison of these alternative results with the necessary and sufficient condition $\Lambda < 1$.

\subsection{Probability One Bounds}

\label{ss:pob}

Several conditions based on probability one bounds have been proposed in the
literature.  A representative example is theorem~3.1
of \cite{epstein1989}, which shows that a solution to the recursive utility
problem exists whenever
\begin{equation}
    \label{eq:pob}
    \psi > 1
    \; \text{ and }  \;
    \beta B_c^{1 - 1/\psi}  < 1,
\end{equation}
where $B_c$ is a probability one upper bound on $C_{t+1} /  C_t$.  If such a
$B_c$ exists, then condition \eqref{eq:pob} is directly comparable with the
condition $\Lambda < 1$ because $B_c$ always exceeds the risk-adjusted mean
consumption growth rate. Indeed, $C_{t+1} / C_t \leq B_c$ for all $t$ implies
$C_n / C_0 \leq B_c^n$ for all $n$, and hence, by the definition of $\lrM$ in
\eqref{eq:drc},
\begin{equation}
    \label{eq:rcvmc}
    \lrM  \leq B_c.
\end{equation}
Thus, condition \eqref{eq:pob} implies $\Lambda < 1$.  In other words, the
condition $\Lambda < 1$ is weaker than the condition of \cite{epstein1989}.
Moreover, unless consumption growth is deterministic, the inequality in
\eqref{eq:rcvmc} is strict.  This fact recalls the point made earlier
that focusing only on the upper tail of the consumption growth
distribution leads to overly pessimistic restrictions.

To illustrate, consider the case of the long-run risk model,
where consumption growth is as specified in \eqref{al:ssyc1}.
Since consumption growth innovations are normally distributed and have no finite upper bound in this case, we have $B_c = +\infty$ and the Epstein--Zin stability
coefficient $\beta B_c^{1 - 1/\psi}$ from condition \eqref{eq:pob} is also
$+\infty$.  In contrast, $\Lambda < 1$ at the baseline parameters, as
discussed in section~\ref{ss:lrr}.

\subsection{Contraction Arguments}

\label{ss:ca}

Recall the specification for consumption used in \cite{alvarez2005using} and
discussed in section~\ref{ss:aj}, where $C_t = \tau^t X_t$ with $\tau > 0$ and
$\{X_t\}$ positive and {\sc iid}.  Let $X_t$ have distribution $F$ on $[a, b]
\subset \RR$ for some positive scalars $a < b$.  When consumption obeys
\eqref{eq:ac}, \cite{alvarez2005using} show that a unique solution to the
recursive utility problem exists whenever
\begin{equation}
    \label{eq:alv}
    \beta \tau^{1 - \frac{1}{\psi}}
    \max_{a \leq x \leq b}
    \left\{
        \int \left(\frac{y}{x} \right)^{1-\gamma}
        \pi(\diff y)
    \right\}^{\frac{1}{\theta}} < 1.
\end{equation}
See \cite{alvarez2005using}, proposition~9 and lemma~A.1 for
details.

By way of comparison, we know from section~\ref{ss:aj} that a unique solution
exists if and only if $\beta \, \tau^{1-\frac{1}{\psi }} < 1$.
This condition is weaker than \eqref{eq:alv}, since the additional maximized term
in \eqref{eq:alv} always exceeds unity.  The difference between the stability
conditions arises because the condition from \cite{alvarez2005using} enforces
contraction in one step.  In contrast, the condition $\Lambda < 1$ is an
asymptotic condition that ignores short-run fluctuations in consumption.  This
leads to a weaker condition because short-run fluctuations do not impinge on
asymptotic outcomes.

\section{Unbounded State Spaces}

\label{s:ub}

Lastly, we return to the setting of theorem~\ref{t:bkl1c} but now dropping the
compactness assumption on the state space (i.e., assumption~\ref{a:c}) and
replacing it with

\begin{assumption}
    \label{a:nc}
     The operator $K$ is continuous and eventually weakly compact.
\end{assumption}

Condition~\ref{a:nc} is similar to assumption~2.1 in
\cite{christensen2017nonparametric}, which is used to study estimation of
positive eigenfunctions of valuation operators.  An explanation of the
terminology and discussion of sufficient conditions can be found in appendix~\ref{s:pi}.

\begin{theorem}
    \label{t:bkgts}
    If assumptions~\ref{a:i} and \ref{a:nc} both hold, then  the following
    statements are equivalent:
    \begin{enumerate}
        \item $\Lambda < 1$.
        \item $A$ has a fixed point in $\cC$.
        \item There exists a $g \in \cC$ such that $\{A^n g\}_{n \geq 1}$
            converges to an element of $\cC$.
    \end{enumerate}
\end{theorem}

The proof of theorem~\ref{t:bkgts} relies heavily on a rather specific spectral continuity
result for positive operators obtained by \cite{schep1980positive}.  This
continuity is combined with an approximation argument that allows us to
bootstrap some results from theorem~\ref{t:bkl1c}.
Assumption~\ref{a:nc} is used to show that the
approximation step is valid.

Compared to theorem~\ref{t:bkl1c}, the calculation of the test
value $\Lambda$ in theorem~\ref{t:bkgts} is more problematic when no analytical solution exists, as even a
flexible method such as Monte Carlo restricts the state space to a finite set of
floating point numbers. Moreover, the conclusions are weaker and the technical
condition in assumption~\ref{a:nc} may be nontrivial to test.  Nevertheless, the theorem
shows us that the core existence result extends to a large class of
unbounded theoretical models.


\section{Concluding Remarks}

\label{s:c}

We derived a simple condition that is both necessary and sufficient for existence and
uniqueness of the continuation value and aggregate wealth-consumption ratio in
recursive utility models with homothetic Epstein--Zin preferences. Despite the
fact that the nonlinear preference recursion intertwines the role of risk
aversion and intertemporal elasticity of substitution, our condition separates
the two forces. What matters is the average long-run risk-adjusted consumption
growth rate and its relationship to intertemporal substitutability of
consumption and time discounting. In asset pricing terms, the condition
imposes a bound on the long-horizon decay rate of the value of consumption
strips as their maturity increases.

We also provided a range of analytical examples, a globally stable method of
computing solutions whenever they exist, and numerical methods that allow for
efficient evaluation of our condition.  Since transitory
details of the consumption process do not affect the average long-run
risk-adjusted consumption growth rate and hence are irrelevant for our
stability condition, insights from this paper can also be applied in
production economy settings without explicitly solving the model, as long as
cointegration properties of consumption with an exogenous driver of
uncertainty can be directly inferred.

While our sharpest results restrict the state space for the underlying Markov state to
be compact, we also give extensions to unbounded state spaces and study
sensitivity of numerical methods as we relax truncation of the state space.
Our results for the unbounded setting are more limited than those for
the compact case, but we conjecture that the combination of ever sharper
numerical methods and the kinds of theoretical insights provided above
forms the most promising road to further understanding of stability and
equilibrium properties of asset pricing models used in modern
quantitative applications.

\appendix

\section{General Fixed Point and Spectral Radius Results}

In what follows, the state space $\XX$ is allowed to be any compact metric
space.  In particular, $\XX$ can be a compact subset of $\RR^n$ with the
Euclidean distance or an arbitrary finite set endowed with the discrete metric
$d(x, y) = \1\{x \not= y\}$.  Note that, in the latter case, $\kappa$ and $q$
are automatically continuous.
The collection of Borel measurable functions $g$
from $\XX$ to $\RR$ such that $\| g \| := \int |g| \diff \pi < \infty$ is denoted
by $L_1(\pi)$.  Convergence is with respect to $\| \cdot \|$ unless
otherwise stated.
For $g, h \in L_1(\pi)$, the statement $g \leq h$
means that $g \leq h$ holds $\pi$-almost everywhere,
while $g \ll h$
indicates that $g < h$ holds $\pi$-almost everywhere.
The symbol $\cC$ denotes all $g \in L_1(\pi)$ such that $g \gg 0$.

Let $M$ be a linear operator mapping $L_1(\pi)$ to itself.  The \emph{operator
norm} and \emph{spectral radius} of $M$  are defined by $\|M\| := \sup
\setntn{\|Mg\|}{g \in L_1(\pi) ,\; \| g \| \leq 1}$ and $r(M) := \lim_{n \to
\infty} \| M^n \|^{1/n}$ respectively.   The operator $M$ is called
\emph{positive} if $Mg \geq 0$ whenever $g \geq 0$.  It is called
\emph{bounded} if $\| M \|$ is finite and \emph{compact} if the image of the
unit ball in $L_1(\pi)$ under $M$ is precompact in
the norm topology.  A (possibly nonlinear) operator $S$ mapping
a convex subset $E$ of $L_1(\pi)$ into itself
is called \emph{convex} on $E$ if $S (\lambda f + (1-\lambda) g) \leq \lambda S
+ (1-\lambda) S g$ for all $f,g \in E$ and all $\lambda \in [0,1]$; and
\emph{concave} if the reverse inequality holds. It is called \emph{isotone} if
$f \leq g$ implies $Sf \leq Sg$.

The following is a local spectral radius result suitable for $L_1(\pi)$.
The proof provided here is due to Miros\l{}awa Zima (private communication)
and draws on earlier results from \cite{zabreiko1967bounds}.

\begin{theorem}[Zabreiko--Krasnosel'skii--Stetsenko--Zima]
    \label{t:lsr}
    Let $h$ be an element of $L_1(\pi)$ and let $M$ be a positive compact
    linear operator.  If $h \gg 0$, then
    \begin{equation}
        \label{eq:lsr}
        \lim_{n \to \infty} \| M^n h \|^{1/n} = r(M).
    \end{equation}
\end{theorem}

\begin{proof}[Proof of theorem~\ref{t:lsr}]
    Let $h$ and $M$ be as in the statement of the theorem.
    Recall that
    $r(h, M) = \limsup_{n \to \infty} \| M^n h \|^{1/n}$ is the local spectral
    radius of $M$ at $h$.  From the definition of $r(M)$ it suffices to show that
    $r(h, M) \geq r(M)$.  To this end, let $\lambda$ be a constant
    satisfying $\lambda > r(h, M)$ and let
    \begin{equation}
        \label{eq:xl}
        h_\lambda := \sum_{n=0}^\infty \frac{M^n h}{\lambda^{n+1}}.
    \end{equation}
    The point $h_\lambda$ is a well-defined element of $L_1(\pi)$ by
        $\limsup_{n \to \infty} \| M^n h \|^{1/n} < \lambda$
    and Cauchy's root test for convergence.  It is also positive $\pi$-almost
    everywhere, since
    the sum in \eqref{eq:xl} includes $h \gg 0$, and since $M$ is a positive
    operator.  Moreover, by standard Neumann series theory (e.g.,
    \cite{krasnosel2012approximate}, theorem~5.1), the point $h_\lambda$
    also has the representation $h_\lambda = (\lambda I - M)^{-1} h$,
    from which we obtain $\lambda h_\lambda - M h_\lambda = h$.  Because
    $h \in \cC$, this implies that
         $M h_\lambda \leq \lambda h_\lambda$.
    Applying this last inequality, compactness of $M$,
    quasi-interiority of $h_\lambda$ and theorem~5.5 (a) of
    \cite{krasnosel2012approximate}, we must have $r(M) \leq \lambda$.
    Since this inequality was established for an arbitrary $\lambda$
    satisfying $\lambda > r(h, M)$, we conclude that $r(h, M) \geq r(M)$.
    Hence $r(h, M) = r(M)$.
    Finally, since $M$ is compact, corollary~1 of \cite{danevs1987local}
    implies that $r(h, M) = \lim_{n \to \infty} \| M^n h \|^{1/n}$, so
    \eqref{eq:lsr} holds.
\end{proof}


The next result is an extension of theorem~\ref{t:lsr}, which weakens the
compactness condition in that theorem while requiring additional positivity.

\begin{theorem}
    \label{t:lsr2}
    Let $h$ be an element of $L_1(\pi)$ and let $M$ be a
    linear operator on $L_1(\pi)$.
        If $M^i$ is compact for some $i \in \NN$ and
        $M f \gg 0$ whenever $f \in L_1(\pi)$ and $f \gg 0$, then
    \begin{equation}
        \label{eq:lsr2}
        \lim_{n \to \infty}
        \left\{ \int M^n h \, \diff \pi \right\}^{1/n} = r(M)
        \quad \text{for all } h \gg 0.
    \end{equation}
\end{theorem}

\begin{proof}
    Fix $h \in L_1(\pi)$ with $h \gg 0$ and
    choose $i \in \NN$ such that $M^i$ is a compact linear operator on
    $L_1(\pi)$.  Fix $j \in \NN$ with $0 \leq j \leq i-1$.
    By our assumptions on $M$ we know that $M^j h \gg 0$, so theorem~\ref{t:lsr} applied to $M^i$
    with initial condition $M^j h$ yields
    \begin{equation*}
        \left\{ \int M^{in} M^j h \, \diff \pi  \right\}^{1/n}
        = \left\{ \int M^{in+j} h \, \diff \pi  \right\}^{1/n}
        \to r(M)^i
        \qquad (n \to \infty).
    \end{equation*}
    But $r(M^i) = r(M)^i$, so
    \begin{equation*}
        \left\{ \int M^{in+j} h \, \diff \pi  \right\}^{1/(in)} \to r(M)
        \qquad (n \to \infty).
    \end{equation*}
    It follows that
    \begin{equation*}
        \left\{ \int M^{in+j} h \, \diff \pi  \right\}^{1/(in+j)} \to
        r(M)
        \qquad (n \to \infty).
    \end{equation*}
    As $j$ is an arbitrary integer satisfying  $0 \leq j \leq i-1$,
    we conclude that \eqref{eq:lsr2} holds.
\end{proof}

The next lemma is useful for detecting fixed points in $L_1(\pi)$.

\begin{lemma}
    \label{l:udfp}
    Let $\{g_n\}$ be a monotone increasing sequence in $L_1(\pi)$.
    If $\{g_n\}$ is bounded above by some $h$ in $L_1(\pi)$, then
    there exists a $g$ in $L_1(\pi)$ such that $g_n \to g$.
    Moreover, if $g_n = T^n g_0$ for some continuous operator $T$
    mapping a subset of $L_1(\pi)$ to itself, then $g$ is a fixed point of $T$.
\end{lemma}

\begin{proof}
    The first claim is a direct consequence of the Monotone Convergence
    Theorem for integrals.  Regarding the second, we have $g_n \to g$ and
    hence, by continuity, $T g_n \to Tg$.  But, by the definition of the sequence $\{g_n\}$,
    we also have $T g_n \to g$.  Hence $Tg=g$.
\end{proof}




\section{Remaining Proofs, Compact Case}

\label{s:pi}

As before, $q$ is the transition density kernel for $\{X_t\}$.  We define
$q^i$ recursively by $q^1 = q$ and $q^i(x, y) = \int q(x, z)
q^{i-1}(z, y) \diff z$ for each $x,y \in \XX$.  By the conditions in
section~\ref{s:cp}, we can take $\ell \in \NN$ such that $q^\ell > 0$,
and $\ell$ has this meaning throughout. Also, in all of what remains, we adopt the notation
\begin{equation}
    \label{eq:defkern}
    k(x, y) =
    \int \exp[ (1-\gamma) \kappa(x, y, \epsilon) ] \nu(\diff \epsilon)
    q(x, y) ,
\end{equation}
so that the operator $K$ satisfies $Kg(x) = \int k(x, y) g(y) \diff y$.
Let $k^i$ be defined on $\XX \times \XX$ by $k^1 = k$ and $k^i(x,
y) = \int k(x, z) k^{i-1} (z, y) \diff z$.  As is easily verified by
induction, the  $i$-th element $k^i$ is the kernel of the $i$-th power of $K$.
That is, for all $x \in \XX$, $i \in \NN$ and $g \in \cC$,
\begin{equation*}
    K^i g(x) = \int k^i(x, y) g(y) \diff y.
\end{equation*}
Throughout this section, assumptions~\ref{a:i} and \ref{a:c} are taken to be
true.

\begin{lemma}
    \label{l:posd}
    The density $\pi$ is the unique stationary density for $q$ on $\XX$.  In
    addition, $\pi$ is everywhere positive and continuous on $\XX$.
\end{lemma}

\begin{proof}
    Regarding uniqueness, the condition that $q^\ell > 0$ for some positive
    $\ell \in \NN$ implies irreducibility of $\{X_t\}$, which in turn implies
    uniqueness of the stationary distribution.  See, e.g.,
    \cite{meyn2009markov}, theorem~10.0.1.\footnote{To verify irreducibility, pick any $A$ with $\int_A \pi(x) \diff x >
    0$.  Evidently $A$ has positive Lebesgue measure, so, with $\ell$ chosen
    so that $q^\ell$ is everywhere positive, we have $\int_A q^\ell(x, y)
    \diff y > 0$ for all $x \in \XX$.  Thus, $A$ is visited with positive
    probability from any state in $\ell$ steps. In other words,
    $\pi$-irreducibility holds.}
    Regarding positivity, suppose to the contrary, there exists a $y$ with
    $\pi(y) = 0$. Since $\pi$ is a stationary density, this means that $\pi(y)
    = \int q^\ell(x, y) \pi(x) \diff x = 0$.  But $q^\ell$ is everywhere positive and
    $\pi$ is a density, and hence positive on a set of positive measure.
    Thus, the integral must be positive. Contradiction.
    Finally, take $y_n \to y$ in $\XX$ and observe that
    \begin{equation*}
        \pi(y_n) = \int q(x, y_n) \pi(x) \diff x
        \to \int q(x, y) \pi(x) \diff x
        = \pi(y)
    \end{equation*}
    as $n \to \infty$.  The convergence of the integrals is due to continuity
    of $q$ and the Dominated Convergence Theorem.  The latter can be employed
    because $q$ is bounded on $\XX \times \XX$ by compactness of $\XX$ and
    continuity of $q$.  Hence $\pi$ is continuous on $\XX$.
\end{proof}

\begin{lemma}
    \label{l:wpok}
    Regarding the operator $K$, the following statements are true:
    \begin{enumerate}
        \item $K$ is a bounded linear operator on $L_1(\pi)$ that maps $\cC$ to itself.
        \item $Kg$ is continuous at each $g \in \cC$.
        \item $Kg \not= 0$ whenever $g \in \cC$ and $g \not= 0$.
        \item $Kg \gg 0$ whenever $g \in \cC$ and $g \gg 0$.
    \end{enumerate}
\end{lemma}

\begin{proof}
    Regarding claim (a), $k$ is continuous and hence bounded by some constant
    $M$ on $\XX$, while $\pi$ is positive and continuous on a compact set, and
    hence bounded below by some positive constant $\delta$.  This yields, for
    arbitrary $f \in L_1(\pi)$,
    \begin{equation}
        \label{eq:kib}
        |Kf(x)| =
        \left| \int k(x, y) f(y) \diff y \right|
        \leq M \int \frac{|f(y)|}{\pi(y)} \pi(y) \diff y
        \leq  \frac{M}{\delta} \|f\|.
    \end{equation}
    It follows directly that $K$ is a bounded linear operator on $L_1(\pi)$.

    Regarding continuity, fix $g \in \cC$, $x \in \XX$ and $x_n \to x$.  Using
    the inequality from \eqref{eq:kib}, we have
    \begin{equation*}
        k(x_n, y) g(y)
        \leq M \frac{g(y)}{\pi(y)} \pi(y)
        \leq  \frac{M}{\delta} g(y)\pi(y).
    \end{equation*}
    Since $g \in L_1(\pi)$, the function on the right hand side is
    integrable with respect to Lebesgue measure, so we can apply the
    Dominated Convergence Theorem to obtain
    \begin{equation*}
        \lim_{n \to \infty} Kg(x_n)
        = \int \lim_{n \to \infty} k(x_n, y) g(y) \diff y
        = Kg(x).
    \end{equation*}
    In particular, $Kg$ is continuous at any $x \in \XX$.

    Regarding claim (c), suppose that, to the contrary, we have
    $Kg = 0$ for some nonzero $g \in \cC$.  Let $B = \{g > 0\}$.
    Since $g$ is nonzero, we have $\pi(B) > 0$.  Since $Kg = 0$,
    it must be the case that $\int_B k(x, y) \diff y = 0$ for any $x \in \XX$.
    But then $\int_B q(x, y) \diff y = 0$ for any $x \in \XX$.
    A simple induction argument shows that this extends to the $n$-step
    kernels, so that, in particular, $\int_B q^\ell (x, y) \diff y
    = 0$ for all $x \in \XX$.  The last equality contradicts $q^\ell > 0$, as guaranteed by
    assumption~\ref{a:i}.

    Part (d) is immediate from $Kg(x) = \int g(y) k(x, y) \diff y$ and \eqref{eq:defkern}.
\end{proof}

The next lemma discusses irreducibility of $K$ as a linear operator on
$L_1(\pi)$.  See p.~262 of \cite{meyer2012banach} for the definition
and further discussion.

\begin{lemma}
    \label{l:irr}
    The operator $K$ is irreducible and $K^2$ is compact.
\end{lemma}

\begin{proof}
    Consider first irreducibility of $K$.
    By positivity of $\kappa$ and the condition for irreducibility of kernel
    operators on p.~262 of \cite{meyer2012banach}, it suffices to show that,
    for any Borel set $A$ of $\XX$ with $0 < \pi(A) < 1$, there exists a Borel
    set $B$ of $\XX$ with
    \begin{equation}
        \label{eq:qti}
        0 < \pi(B) < 1
        \quad \text{and} \quad
        \int_B \int_A q(x, y) \diff x \diff y > 0.
    \end{equation}
    To verify \eqref{eq:qti}, suppose to the contrary that we can choose
    $A \subset \XX$ with $0 < \pi(A) < 1$ and
    \begin{equation*}
        \int_{A^c} \int_A q(x, y) \diff x \diff y
        = \int_A \left[ \int_{A^c} q(x, y) \diff y \right] \diff x
        = 0.
    \end{equation*}

    Then $\int_A  \int_{A^c} q(x, y) \diff y \pi(x) \diff x = 0$.  Hence $A$
    is absorbing for $q$, in which case $\pi(A) = 1$ by $\pi$-irreducibility
    of $q$ and proposition~4.2.3 of \cite{meyn2009markov}.  Contradiction.

    Regarding compactness, theorem~9.9 of \cite{schaeferbanach} implies that
    $K^2$ will be compact whenever $K$ is weakly compact, which requires that the
    image of the unit ball $B_1$ in $L_1(\pi)$ under $K$ is relatively compact
    in the weak topology.  To prove this it suffices to to show that, given
    $\epsilon > 0$, there exists a $\delta > 0$ such that $\int_A K|f| \diff
    \pi < \epsilon$ whenever $f \in B_1$ and $\pi(A) < \delta$.  This is true
    because $k$ is continuous and hence bounded on $\XX$, yielding
    \begin{equation*}
        \int_A \int k(x, y) |f(y)| \diff y \pi(x) \diff x
        = \int \int_A k(x, y) \pi(x) \diff x |f(y)| \diff y
        \leq M \pi(A) \int |f(y)| \diff y
    \end{equation*}
    for some constant $M$.  By continuity and positivity of $\pi$ on $\XX$
    we have $\pi \geq a$ for some constant $a > 0$, and hence $\int |f(y)|
    \diff y \leq a \| f \| \leq a$.  Hence $\int_A K|f| \diff \pi < \epsilon$
    whenever $\pi(A) < \epsilon / (a M)$.
\end{proof}

\begin{proposition}
    \label{p:lsr}
    $\Lambda$ is well defined and satisfies $\Lambda = \beta \, r(K)^{1/\theta}$.
\end{proposition}

\begin{proof}
    Since $K^i$ is compact for some $i$ and maps positive functions into
    positive functions (see lemmas~\ref{l:irr} and~\ref{l:wpok}), we can apply theorem~\ref{t:lsr2} to
    $\1 \equiv 1$ to obtain $r(K) = \lim_{n \to \infty} \| K^n \1 \|^{1/n}$.
    An inductive argument based on \eqref{eq:kappa} shows that, for each $n$
    in $\NN$, we have $K^n \1 (x)
        = \EE_x \left( C_n/C_0 \right)^{1-\gamma}$.
    Hence, by the law of iterated expectations,
    \begin{equation}
        \label{eq:lim1}
        \| K^n \1 \|^{1/n}
        =
        \
        \left\{
            \EE \left( \frac{C_n}{C_0} \right)^{1-\gamma}
        \right\}^{1/n}
        \\
        =
        \left\{
            \rR \left( \frac{C_n}{C_0} \right)
        \right\}^{(1-\gamma)/n}.
    \end{equation}
    Since $r(K) = \lim_{n \to \infty} \| K^n \1 \|^{1/n}$, this yields
    \begin{equation*}
        r(K)
         =
        \lim_{n \to \infty} \,
        \left\{
            \rR \left( \frac{C_n}{C_0} \right)
        \right\}^{(1-\gamma)/n}
        =
        \lrM^{1 - \gamma}.
    \end{equation*}
    Because $\theta := (1-\gamma)/(1 - 1/\psi)$, we now have
    \begin{equation*}
        \beta r(K)^{1/\theta}
        = \beta \lrM^{1 - 1/\psi}
        = \Lambda.
        \qedhere
    \end{equation*}
\end{proof}

\begin{theorem}
    \label{t:sri}
    The spectral radius $r(K)$ of $K$ is strictly positive.  Moreover, there
    exists an eigenfunction $e$ of $K$ satisfying
    \begin{equation}
        \label{eq:pie}
        K e = r(K) e
        \quad \text{and} \quad
        e \gg 0.
    \end{equation}
    The function $e$ is continuous everywhere on $\XX$.
\end{theorem}

\begin{proof}
    Theorem~4.1.4 and lemma~4.2.9 of \cite{meyer2012banach} together with the
    irreducibility and compactness properties of $K$ obtained in
    lemma~\ref{l:irr} yield positivity of $r(K)$ and existence of the positive
    eigenfunction in \eqref{eq:pie}.  Lemma~\ref{l:wpok} implies that $e$ is
    continuous, since $e \in \cC$ and $e = (K e) / r(K)$.
\end{proof}

In what follows, $e$ in \eqref{eq:pie} is called the Perron--Frobenius
eigenfunction.

Since $e$ is positive and continuous, the constants $\underline e := \min_{x
\in \XX} e(x)$ and $\bar e := \max_{x \in \XX} e(x)$ are finite and strictly
positive.  These facts are now used to study $A$ from \eqref{eq:defa}.

\begin{lemma}
    \label{l:eb}
    Let $e$ be the Perron--Frobenius eigenfunction of $K$.  If
    \begin{equation}
        \label{eq:cpk}
        \lim_{t \downarrow 0} \frac{\phi(t)}{t} \, r(K) > 1
        \quad \text{and} \quad
        \lim_{t \uparrow \infty} \frac{\phi(t)}{t} \, r(K) < 1,
    \end{equation}
    then there exist positive constants $c_1 < c_2$ with the following
    properties:
    \begin{enumerate}
        \item If $\, 0 < c \leq c_1$ and $f = ce$, then there exists a
            $\delta_1
            > 1$ such that $Af \geq \delta_1 f$
        \item If $\, c_2 \leq c < \infty$ and $f = ce$, then  there exists a
            $\delta_2 < 1$ such that $Af \leq \delta_2 f$.
    \end{enumerate}
\end{lemma}

\begin{proof}
    Let $\lambda := r(K)$ and let $e$ be the Perron--Frobenius eigenfunction.
    Let $\underline e$ and $\bar e$ be the maximum and minimum of $e$ on
    $\XX$, as defined above.
    Regarding claim (a), observe that, in view of \eqref{eq:cpk}, there exists
    a $\delta_1 > 1$ and an $\epsilon > 0$ such that
    \begin{equation*}
        \frac{\phi(t)}{t} \lambda \geq \delta_1
        \quad \text{whenever} \quad 0 < t < \epsilon.
    \end{equation*}
    Choosing $c_1$ such that $0 < c_1 \lambda \bar e < \epsilon$ and $c \leq c_1$, we have
    $c \lambda e(x) < \epsilon$ for all $x \in \XX$, and hence
    \begin{equation*}
        A ce(x)
        = \phi(c Ke (x))
        = \phi(c \lambda e (x))
        = \frac{ \phi(c \lambda e (x)) }{c \lambda e (x)} c \lambda e (x)
        \geq \delta_1 c e (x).
    \end{equation*}

    Turning to claim (b) and using again the hypotheses in
    \eqref{eq:cpk}, we can choose a $\delta_2 <1$ and finite constant $M$ such that
    \begin{equation*}
        \frac{\phi(t)}{t} \lambda \leq \delta_2
        \quad \text{whenever} \quad t > M.
    \end{equation*}
    Let $c_2$ be a constant strictly greater than
    $\max\{M/(\lambda \underline e), c_1\}$ and fix $c \geq c_2$.  By the
        definition of $\underline e$ we have $c \lambda e(x) \geq c_2 \lambda
        \underline e > M$
        for all $x \in \XX$, so
    \begin{equation*}
        A c e (x)
        = \phi(c \lambda e (x))
        = \frac{ \phi(c \lambda e (x)) }{c \lambda e (x)} \lambda c e (x)
        \leq \delta_2 c e (x).
    \end{equation*}
    By construction, $0 < c_1 < c_2$, so all claims are now established.
\end{proof}

\begin{lemma}
    \label{l:rk21}
    If the conditions in \eqref{eq:cpk} hold and $A$ has a fixed point $g^*$
    in $\cC$, then, given any $g\in \cC$, there exist functions $f_1, f_2 \in
    \cC$ such that
    \begin{equation}
        \label{eq:egb}
        f_1 \leq A g, g^* \leq f_2,
        \quad
        Af_1 \geq f_1 + \epsilon (f_2 - f_1)
        \quad \text{and} \quad
        Af_2 \leq f_2 - \epsilon (f_2 - f_1).
    \end{equation}
\end{lemma}

\begin{proof}
    Fix $g \in \cC$. Since $A g$ is continuous and
    $\XX$ is compact, $Ag$ attains a finite maximum and strictly positive
    minimum on $\XX$.  The same is true of the fixed point $g^* = A g^*$ and
    the Perron--Frobenius eigenfunction $e$.  Hence, we can choose constants
    $a_1$ and $a_2$ such that $0 \ll a_1 e \leq g^*, A g \leq a_2 e$.  With $a_1$
    chosen sufficiently small lemma~\ref{l:eb} implies that $A(a_1 e)
    \geq \delta a_1 e$ for some $\delta > 1$.
    Setting $f_i := a_i e$, we then have $Af_1 \geq \delta a_1 e$.
    Since $\delta > 1$, we can write this as
    $Af_1 \geq a_1 e + \epsilon (a_2 - a_1) e$ for some positive $\epsilon$.
    In other words, $Af_1 \geq f_1 + \epsilon (f_2 - f_1)$.
    The proof of the last inequality is similar.
\end{proof}

\begin{theorem}
    \label{t:bkb}
    If $\Lambda < 1$, then $A$ is globally stable on $\cC$.
\end{theorem}

\begin{proof}
    First we show that, if $\Lambda < 1$, then the conditions in
    \eqref{eq:cpk} hold.  Throughout we use the fact that $\Lambda = \beta
    r(K)^{1/\theta}$, as shown in proposition~\ref{p:lsr}.  To start, observe that
    \begin{equation}
        \label{eq:rat}
        \frac{\phi(t)}{t}
        = \left\{ \frac{1 - \beta}{t^{1/\theta}} + \beta \right\}^\theta .
    \end{equation}
    If, on one hand, $\theta < 0$, then $\Lambda < 1$ implies $\beta^\theta
    r(K) > 1$ and, in addition, \eqref{eq:rat}
    increases to $\beta^{\theta}$ as $t \to 0$.  Thus, the first inequality in
    \eqref{eq:cpk} holds.  The second inequality also holds because $\phi(t)/t
    \to 0$ as $t \to \infty$.  If, on the other hand, $\theta > 0$, then
    $\beta^\theta r(K)
    < 1$ and  \eqref{eq:rat} diverges to $+\infty$ as $t \to 0$, so the
    first inequality in \eqref{eq:cpk} holds.  The second inequality also
holds because $\phi(t)/t \to \beta^\theta$ as $t \to \infty$.

    To complete the proof of theorem~\ref{t:bkb}, note that $\phi$ is either
    concave or convex, depending on the value of $\theta$. Suppose first that
    $\phi$ is concave, which implies that $A$ is both isotone and concave on $\cC$.
    Lemma~\ref{l:eb} yields positive constants $c_1 < c_2$ such that $A c_1 e
    \geq c_1 e$ and $A c_2 e \leq c_2 e$.  Theorem~2.1.2 of \cite{zhang2013},
    which in turn is based on \cite{du1990fixed},
    now implies that $A$ has a fixed point $g^* \in L_1(\pi)$ satisfying $c_1
    e \leq g^* \leq c_2 e$.  Since $e \gg 0$ and $c_1 > 0$, we have $g^* \gg
    0$.

    Let $g$ be a nonzero element of $\cC$.  Choose $f_1, f_2$ as in
    lemma~\ref{l:rk21}.  Theorem~2.1.2 of \cite{zhang2013} now implies that
    every element of $[f_1, f_2]$ converges to $g^*$ under iteration of $A$.
    In particular, $A^n (A g) \to g^*$ as $n \to \infty$.  But then $A^n g \to
    g^*$ also holds.  We conclude that $A$ is globally asymptotically stable
    on $\cC$.

    The proof for the convex case is essentially identical.
\end{proof}

The next result expands on a line of argument developed by
\cite{toda2018wealth}, shifting up to infinite dimensions and allowing $\theta
< 0$.

\begin{proposition}
    \label{p:bkn}
    If $A$ has a nonzero fixed point in $\cC$, then $\Lambda < 1$.
\end{proposition}

\begin{proof}
    Let $K^*$ be the adjoint operator associated with $K$.
    Since $K$ is irreducible and $K^2$ is compact,
    we can employ the version of the Krein--Rutman theorem presented in
    lemma~4.2.11 of \cite{meyer2012banach}.  Combining this result with the
    Riesz Representation Theorem, there exists an $e^* \in L_\infty(\pi)$ such
    that
    \begin{equation}
        \label{eq:pao}
        e^* \gg 0
        \quad \text{and} \quad
        K^* e^* = r(K) e^*.
    \end{equation}
    Let $g$ be a nonzero fixed point of $A$ in $\cC$.  It is convenient in
    what follows to use the inner product notation $\la f, h \ra := \int f h
    \diff \pi$ for $f \in L_1(\pi)$ and $g \in L_\infty(\pi)$.

    First consider the case where $\theta < 0$, so that, by the definition of
    $\phi$ we have $\phi(t) \leq \beta^\theta t$ with strict inequality
    whenever $t > 0$.  As a result, $g(x) = Ag(x) = \phi(K g(x)) \leq
    \beta^\theta K g(x)$, with strict inequality whenever $Kg(x) > 0$.  By
    part (c) of lemma~\ref{l:wpok}, we have $Kg \geq 0$ and $Kg \not= 0$.  So
    it must be that $g \leq \beta^\theta Kg$ and $g < \beta^\theta K g$ on a
    set of positive $\pi$-measure.  But then, taking $e^*$ as in
    \eqref{eq:pao}, we have $\la e^*, \beta^\theta Kg -g \ra > 0$, or,
    equivalently, $\beta^\theta \la e^*, Kg \ra > \la e^*, g \ra$.  Using the
    definition of the adjoint and \eqref{eq:pao} gives
        $r (K) \la e^*, g \ra
        = \la K^* e^*, g \ra
        = \la e^*, Kg \ra$, so it must be that $\beta^\theta r (K) \la e^*, g
        \ra > \la e^*, g \ra$.  Hence $\Lambda = \beta^\theta r (K) > 1$.
        Because $\theta < 0$, this implies that $\Lambda = \beta
        r(K)^{1/\theta} < 1$.

    Next consider the case where $\theta > 0$, so that $\phi(t) > \beta^\theta t$ whenever
    $t > 0$.
    As a result, we have $g(x) = Ag(x) = \phi(K g(x)) \geq \beta^\theta K g(x)$, with
    strict inequality whenever $Kg(x) > 0$.
    By part (c) of lemma~\ref{l:wpok}, we have $Kg \geq 0$ and $Kg \not= 0$.
    So it must be that $g \geq \beta^\theta Kg$ and $g > \beta^\theta K g$ on a set of positive $\pi$-measure.
    But then, taking $e^*$ as in \eqref{eq:pao}, we have
    $\la e^*, \beta^\theta Kg -g \ra < 0$, or, equivalently,
    $\beta^\theta \la e^*, Kg \ra < \la e^*, g \ra$.  As already shown, we
    have $r (K) \la e^*, g \ra = \la e^*, Kg \ra$,
    so it must be that $\beta^\theta r (K) \la e^*, g \ra < \la e^*, g \ra$.
    Hence $\beta^\theta r (K) < 1$.
    Because $\theta > 0$, this implies that $\Lambda = \beta r(K)^{1/\theta} < 1$.
\end{proof}

\begin{proof}[Proof of theorem~\ref{t:bkl1c}]
    \label{pr:bkl1c}
    Clearly (e) $\implies$ (d), which in turn implies (c) since we can take
    $g$ equal to the fixed point.  In addition, (c) implies (b), since
    $K$ is a bounded linear operator on $L_1(\pi)$ and $\phi$ is continuous
    on $\RR_+$, from which it follows that $A$ is continuous on $\cC$, and
    hence any limit of a sequence of iterates $\{A^n g\}_{n \geq 1}$ of $A$ is
    a fixed point of $A$.
    The implication (b) $\implies$ (a) is due to proposition~\ref{p:bkn}.
    Finally, (a) $\implies$ (e) by theorem~\ref{t:bkb} and
    proposition~\ref{p:lsr}.
\end{proof}

\bibliographystyle{apa}

\bibliography{sru_bib}

\end{document}